\documentclass[11pt,authoryear]{article}
\usepackage[paperheight=11.7in,paperwidth=8.3in]{geometry}

\usepackage[utf8]{inputenc}
\usepackage{amsmath}
\usepackage{amssymb}
\usepackage{amsfonts}
\usepackage{amsthm}
\usepackage{bm}
\usepackage{nicefrac}

\usepackage{appendix}
\usepackage{enumerate}
\usepackage{enumitem}
\usepackage{nicefrac}
\usepackage{array}
\usepackage{calc}
\usepackage{flafter}
\usepackage{esdiff}
\usepackage{mathabx}
\usepackage{textcomp}
\usepackage{color}		
\usepackage{graphicx}
\usepackage{amscd}
\usepackage{mathrsfs}
\usepackage{enumerate}
\usepackage{mdwlist}
\usepackage[round]{natbib}
\usepackage{authblk}
\usepackage[colorlinks=true,linkcolor=blue,citecolor=blue,urlcolor=blue]{hyperref}
\usepackage[export]{adjustbox}

\usepackage{algorithm}
\usepackage[noend]{algpseudocode}

\usepackage{caption}
\usepackage{subcaption}

\hypersetup{
    colorlinks,
    citecolor=black,
    filecolor=black,
    linkcolor=black,
    urlcolor=black
}

\makeatletter
\def\BState{\State\hskip-\ALG@thistlm}
\makeatother

\makeatletter
\def\step{%
   \@ifnextchar[ \@myitem{\@noitemargtrue\@myitem[\@itemlabel]}}
\def\@myitem[#1]{\item[#1]\mbox{}}
\makeatother

\newcommand\R{\mathbb{R}}

\newcommand\E{\mathbb{E}}

\newcommand{\calC}{\mathcal{C}}
\newcommand{\calE}{\mathcal{E}}
\newcommand{\calF}{\mathcal{F}}

\newcommand{\calK}{\mathcal{K}}

\newcommand{\calM}{\mathcal{M}}

\newcommand{\calS}{\mathcal{S}}
\newcommand{\calT}{\mathcal{T}}

\newcommand{\ra}{\rightarrow}
\newcommand{\la}{\leftarrow}

\renewcommand\epsilon{\varepsilon}

\newcommand{\vect}[1]{\mathbf{#1}}

\DeclareMathOperator*{\argmin}{arg\,min}

\newcommand{\M}{\calM}

\newtheorem{proposition}{Proposition}

\title{Representation and reconstruction of covariance operators in linear inverse problems}
\author[1]{Eardi Lila\thanks{elila@uw.edu}}
\author[2]{Simon Arridge\thanks{s.arridge@cs.ucl.ac.uk}}
\author[3]{John A. D. Aston\thanks{j.aston@statslab.cam.ac.uk}}

\affil[1]{Department of Biostatistics, University of Washington}
\affil[2]{Centre for Medical Image Computing, University College London}
\affil[3]{Statistical Laboratory, DPMMS, University of Cambridge}

\date{}

\begin{document}

\maketitle

\begin{abstract}
We introduce a framework for the reconstruction and representation of functions in a setting where these objects cannot be directly observed, but only indirect and noisy measurements are available, namely an inverse problem setting. The proposed methodology can be applied either to the analysis of indirectly observed functional images or to the associated covariance operators, representing second-order information, and thus lying on a non-Euclidean space. To deal with the ill-posedness of the inverse problem, we exploit the spatial structure of the sample data by introducing a flexible regularizing term embedded in the model. Thanks to its efficiency, the proposed model is applied to MEG data, leading to a novel approach to the investigation of functional connectivity.
\end{abstract}

\section{Introduction}\label{sec:intro}
An inverse problem is the process of recovering missing information from indirect and noisy observations. Not surprisingly, inverse problems play a central role in numerous fields such as, to name a few, geophysics \citep{Zhdanov2002}, computer vision \citep{Hartley2004}, medical imaging \citep{Arridge1999, Lustig2008} and machine learning \citep{DeVito2005}.

Solving a linear inverse problem means finding an unknown $x$, for instance a function or a surface, from a noisy observation $y$, which is a solution to the model
\begin{equation}\label{eq:inv_model_intro}
y = \calK x + \epsilon,
\end{equation}
where $y$ and $\epsilon$ belong to an either finite or infinite dimensional Banach space. The map $\calK$ is called forward operator and is generally assumed to be known, although its uncertainty has also been taken into account in the literature \citep{Arridge2006, Golub1980, Gutta2019, Kluth2017, Lehikoinen2007, Nissinen2009, Zhu2011}. The term $\epsilon$ represents observational error.

Problem \ref{eq:inv_model_intro} is a well-studied problem within applied mathematics (for early works in the field, see \cite{Calderon1980, Geman1990, Adorf1995}). Its main difficulties arise from the fact that, in practical situations, an inverse of the forward operator does not exist, or if it does, it amplifies the noise term. For this reason such a problem is called ill-posed. Consequently, the estimation of the function $x$ in (\ref{eq:inv_model_intro}) is generally tackled by minimizing a functional which is the sum of a data (fidelity) term and a regularizing term encoding prior information on the function to be recovered \citep[see, among others, ][]{Tenorio2001,Mathe2006,Cavalier2008, Lefkimmiatis2012, YueHu2012}. For convex optimization functionals, modern efficient optimization methods can be applied \citep{Boyd2010, Beck2009, Chambolle2011, Chambolle2016, Burger2016}. Alternatively, when it is important to assess the uncertainty associated with the estimates, a Bayesian approach could be adopted \citep{Kaipio2005, Calvetti2007,Stuart2010, Repetti2019}. The deep convolutional neural network approach has also been applied to this setting \citep{Jin2017,McCann2017}.

In imaging sciences, it is sometimes of interest to find an optimal representation and perform statistics on the second order information associated with the functional samples, i.e. the covariance operators describing the variability of the underlying functional images. This is, for instance, the case in a number of areas of neuroimaging, particularly those investigating functional connectivity. In this work, we establish a framework for reconstructing and optimally representing indirectly observed samples $\calC_1,\ldots, \calC_n$, that are covariance operators, expressing the second order properties of the underlying unobserved functions. The indirect observations are covariance operators generated by the model
\begin{equation}\label{eq:cov_model_intro}
\calS_i = \calK_i \circ \calC_i \circ \calK_i^* + \calE_i, \qquad i=1,\ldots,n,
\end{equation}
where $\calK_i^*$ denotes the adjoint operator and the term $\calE_i$ models observational error. The term $\calK_i \circ \calC_i \circ \calK_i^*$ represents the covariance operator of $\calK_i X^{(i)}$, with $X^{(i)}$ an underlying random function whose covariance operator is $\calC_i$. 

As opposed to more classical linear inverse problems formulations, Problem \ref{eq:cov_model_intro} introduces the following additional difficulties:
\begin{itemize}
\item We are in a setting where each sample is a high-dimensional object that is a covariance operator; it is important to take advantage of the information from all the samples to reconstruct and represent each of them. 
\item The elements $\{\calC_i\}$ and $\{\calS_i\}$ live on non-Euclidean spaces, as they belong to the positive semidefinite cone, and it is important to account for this manifold structure in the formulation of the associated estimators.
\item In an inverse problem setting it is fundamental to be able to introduce spatial regularization, however it is not obvious how to feasibly construct a regularizing term for covariance operators reflecting, for instance, smoothness assumptions on the underlying functional images.
\end{itemize}

More general non-Euclidean settings could also be accommodated. Specifically, the error term could be defined on a tangent space and mapped to the original space through the exponential mapping. Another setting of interest is the case of error terms that push the observables out of the original space. In our applications this is not an issue, as the contaminated observations are themselves empirical covariance matrices, which belong to the non-Euclidean space of positive semidefinite matrices.

We tackle Problem \ref{eq:cov_model_intro} by generalizing the concept of Principal Component Analysis (PCA) to optimally represent and understand the variation associated with samples that are indirectly observed covariance operators. The proposed model is also able to deal with the simpler case of samples that are indirectly observed functional images belonging to a linear functional space.

\subsection{Motivating application - functional connectivity}
In recent years, statistical analysis of covariance matrices has gained a predominant role in medical imaging and in particular in functional neuroimaging. In fact, covariance matrices are the natural objects to represent the brain's functional connectivity, which can be defined as a measure of covariation, in time, of the cerebral activity among brain regions. While many techniques have been proposed to describe functional connectivity, almost all can be described in terms of a function of a covariance or related matrix.

\begin{figure}[!htb] 
\centering
\includegraphics[width=0.6\textwidth]{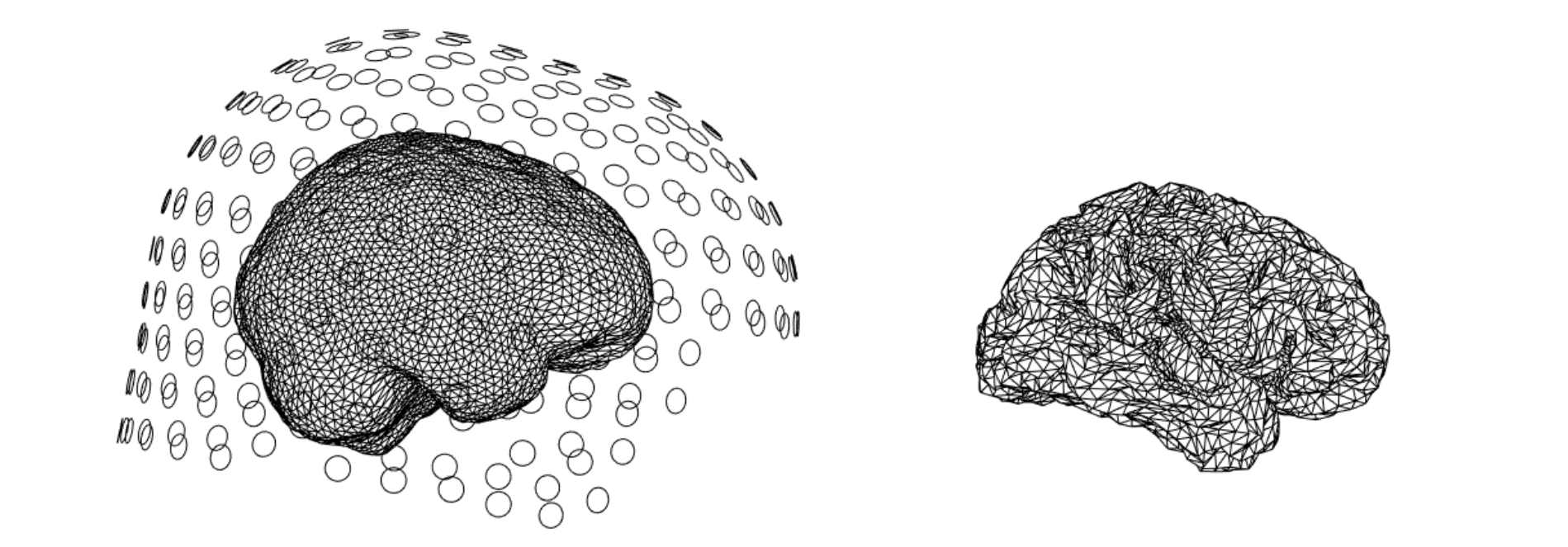}
\includegraphics[width=0.62\textwidth]{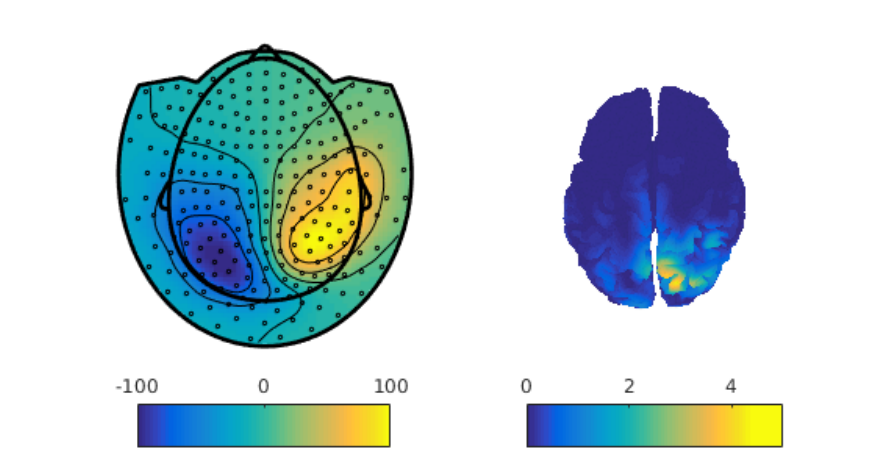}
\caption[]{On the top left, head model of a subject and superimposition of the 248 MEG sensors positioned around the head, called `sensors space'. On the top right, brain model of the same subject represented by a triangular mesh of 8K nodes, which represents the `brain space'. On the bottom left, an example of a synthetic signal detected by the MEG sensors. The dots represent the sensors, the color map represents the signal detected by the sensors. On the bottom right, intensity of the reconstructed signal on the triangular mesh of the cerebral cortex.}
\label{fig:forwardoperator}
\end{figure}

\begin{figure}[!htb] 
\includegraphics[width=\textwidth]{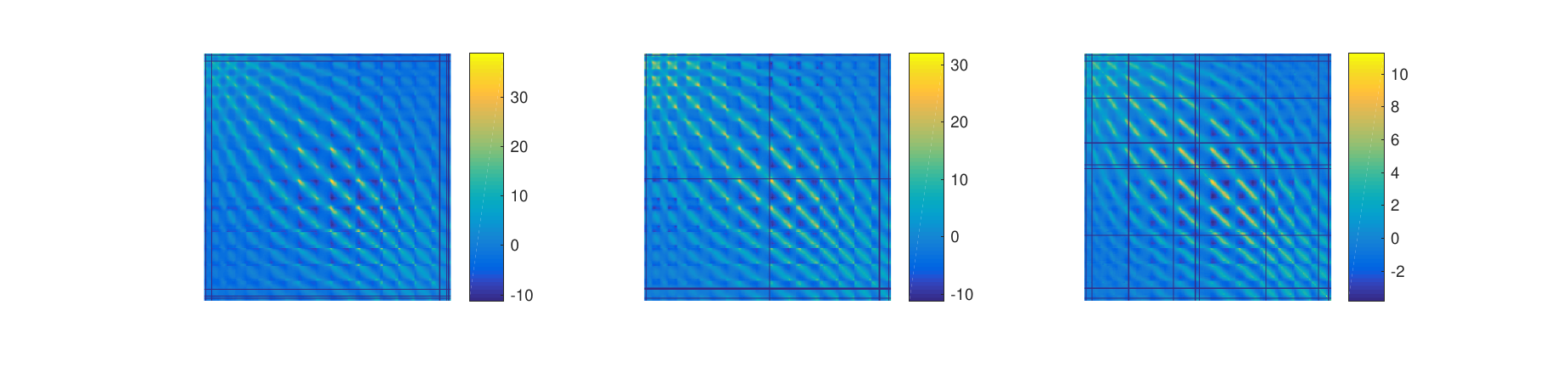}
\caption[]{Covariance matrices of the signal detected by the MEG sensors from three different subjects of the Human Connectome Project. The size of the matrices is $248 \times 248$. The dark blue bands represent missing data, which are due to the exclusion of some channels after a quality check of the signal.}
\label{fig:covariances}
\end{figure}

Covariance matrices representing functional connectivity can be computed from the signals arising from functional imaging modalities. The choice of a specific functional imaging modality is generally driven by the preference to have high spatial resolution signals, and thus high spatial resolution covariance matrices, versus high temporal resolution, and thus the possibility to study the temporal dynamic of the covariance matrices. Functional Magnetic Resonance falls in the first category, while  Electroencephalogram (EEG) and Magnetoencephalography (MEG) in the second. However, high temporal resolution does generally come at the price of indirect measurements and, as shown in Figure~\ref{fig:forwardoperator} in the case of MEG data, the signals are in practice detected on the \textit{sensors space}. It is however of interest to produce results on the associated signals on the cerebral cortex, which we will refer to as \textit{brain space}. The signals on the brain space are functional images whose domain is the geometric representation of the brain and are associated with the neuronal activity on the cerebral cortex. We borrow here the notion of brain space and sensors space from \cite{johnstone1990} and we use it throughout the paper for convenience, however it is important to highlight that the formulation of the problem is much more general than the setting of this specific application.

The signals on the brain space are related to the signals on the sensors space by a forward operator, derived from the physical modeling of the electrical/magnetic propagation, from the cerebral cortex to the sensors. This is generally referred to as the forward problem. For soft-field methods like EEG, MEG and Functional Near-Infrared Spectroscopy \citep{Mosher1999, Eggebrecht2014, Ferrari2012, Singh2014, Ye2009}, the forward operator is defined through the solution to a partial differential equation of diffusion type. Such a mapping induces a strong degree of smoothing and consequently the corresponding inverse problem, i.e. the reconstruction of a signal on the brain space from observations in the sensors space, is strongly ill-posed. In fact, signals with fairly different intensities on the brain space, due to the diffusion effect, result in signals with similar intensities in the sensors space. In Figure~\ref{fig:forwardoperator}, we show an example of a signal on the brain space and the associated signal on the sensors space.

From a practical perspective, it is crucial to understand how the different parts of the brain interact, which is sometimes known as functional connectivity. A possible way to understand these interactions is by analyzing the covariance function associated with the signals describing the cerebral activity of an individual on the brain space \citep{Fransson2011, Lee2013, Li2009}. More recently, the interest has shifted from this static approach to a dynamic approach. In particular, for a single individual, it is of interest to understand how these covariance functions vary in time. This is a particularly active field, known as dynamic functional connectivity \citep{Hutchison2013}. Another element of interest is understanding how these covariance functions vary among individuals. In Figure~\ref{fig:covariances}, we show the covariance matrices, on the sensors space, computed from the MEG signals of three different subjects.

The remainder of this paper is organized as follows. In Section~\ref{sec:intro_math} we give a formal description of the problem. We then introduce a model for indirectly observed smooth functional images in Section~\ref{sec:PCfunctions} and present the more general model associated with Problem~\ref{eq:cov_model_intro} in Section~\ref{sec:PCcovariances}. In Section~\ref{sec:simulations}, we perform simulations to assess the validity of the estimation framework. In Section~\ref{sec:application} we apply the proposed models to MEG data and we finally give some concluding remarks in Section~\ref{sec:discussion}.

\section{Mathematical description of the problem}\label{sec:intro_math}
We now introduce the problem using our driving application as an example. To this purpose, let $\M$ a be a closed smooth two-dimensional manifold embedded in $\R^3$, which in our application represents the geometry of the cerebral cortex. An example of such a surface is shown on the top right of Figure~\ref{fig:forwardoperator}. We denote with $L^2(\M)$ the space of square integrable functions on $\M$. Define $X$ to be a random function with values in a Hilbert functional space $\calF \subset L^2(\M)$ with mean $\mu = \E[X]$, finite second moment, and assume the continuity and square integrability of its covariance function $C_X(v, v') = \E[(X(v)-\mu(v))(X(v')-\mu(v'))]$. The associated covariance operator $\calC_X$ is defined as $\calC_X g = \int_\M C_X(v,v')g(v)dv$, for all $g \in L^2(\M)$.
Mercer’s Lemma \citep{riesz1955} guarantees the existence of a non-increasing sequence $\{\gamma_r\}$ of eigenvalues of $\calC_X$ and an orthonormal sequence of corresponding eigenfunctions $\{\psi_r\}$, such that
\begin{equation}\label{eq:cov_eigen}
C_X(v,v') = \sum_{r=1}^{\infty} \gamma_r \psi_r(v)\psi_r(v'), \qquad  \forall v,v' \in \M.
\end{equation}

As a direct consequence, $X$ can be expanded\footnote{More precisely, we have that $\lim_{R \ra \infty} \sup_{v \in \M} \E \{X(v) - \mu(v) - \sum_{r=1}^{R} \zeta_{r} \psi_r(v)\}^2 = 0$, i.e. the series converges uniformly in mean-square.} as $X = \mu + \sum_{r=1}^{\infty} \zeta_r \psi_r$, where the random variables $\{\zeta_r\}$ are uncorrelated and are given by $\zeta_r = \int_\M \{ X(v) - \mu(v) \} \psi_r(v)dv$. The collection $\{\psi_r\}$ defines the modes of variation of the random function $X$, in descending order of strength, and these are called Principal Component (PC) functions. The associated random variables $\{\zeta_r\}$ are called PC scores. Moreover, the defined PC functions are the best finite basis approximation in the $L^2$-sense, therefore for any fixed $R \in \mathbb{N}$, the first $R$ PC functions of $X$ minimize the reconstruction error, i.e.
\begin{equation}\label{eq:fPCA_minimization}
{\{\psi_r\}}_{r=1}^R =
\argmin_{(\{\phi_r\}_{r=1}^R: \langle \phi_r, \phi_{r'} \rangle = \delta_{rr'})}
\mathbb{E}\int_{\mathcal{M}}\bigg\{X(v) - \mu(v) - \sum_{r=1}^R \langle X-\mu, \phi_r \rangle \phi_r(v) \bigg\}^2 dv,
\end{equation}
where $\langle \cdot, \cdot \rangle$ denotes the $L^2(\M)$ inner product and $\delta_{rr'}$ is the Kronecker delta; i.e. $\delta_{rr'}=1$ for $r=r'$ and $0$ otherwise. 

\subsection{The case of indirectly observed functions}\label{sec:ismpca_intro}
In the case of indirect observations, the signal is detectable only through $s$ sensors on the sensors space. Let $\{K_l: l=1,\ldots,m\}$ be a collection of $s \times p$ real matrices, representing the potentially sample-specific forward operators relating the signal at $p$ pre-defined points $\{v_j: j=1,\ldots,p\}$ on the cortical surface $\M$ with the signal captured by the $s$ sensors. The matrices $\{K_l\}$ are discrete versions of the forward operator $\calK$ introduced in Section~\ref{sec:intro}. Moreover, define the evaluation operator $\Psi: \calF \ra \R^{p}$ to be a vector-valued functional that evaluates a function $f \in \calF$ at the $p$ pre-specified points $\{v_j\} \subset \M$, returning the $p$ dimensional vector $(f(v_1), \ldots, f(v_p))^T$.  The operators $\Psi$ and $\{K_l\}$ are known. However, in the described problem the random function $X$ can be observed  only through indirect measurements $\{\vect{y}_l \in \R^s: l = 1,\ldots,m \}$ generated from the model
\begin{align}\label{eq:model}
\begin{cases}
x_l &= \mu + \sum_{r=1}^{\infty} \zeta_{l,r} \psi_r\\
\vect{y}_l &= K_l \Psi x_l + \bm{\epsilon}_l, \qquad l=1,\ldots,m
\end{cases}
\end{align}
where $\{x_l\}$ are $m$  independent realizations of $X$, and thus expandible in terms of the PC functions $\{\psi_r\}$ and the coefficients $\{\zeta_{l,r}\}$ given by $\zeta_{l,r} = \int_\M \{ x_l(v) - \mu(v) \} \psi_r(v)dv$. The terms $\{\bm{\epsilon}_l\}$ represent observational errors and are independent realizations of a $s$-dimensional normal random vector, with mean the zero vector and variance $\sigma^2 I_{p}$, where $I_{p}$ denotes the $p$-dimensional identity matrix.

\begin{figure}[!htb] 
\centering
\includegraphics[width=0.6\textwidth]{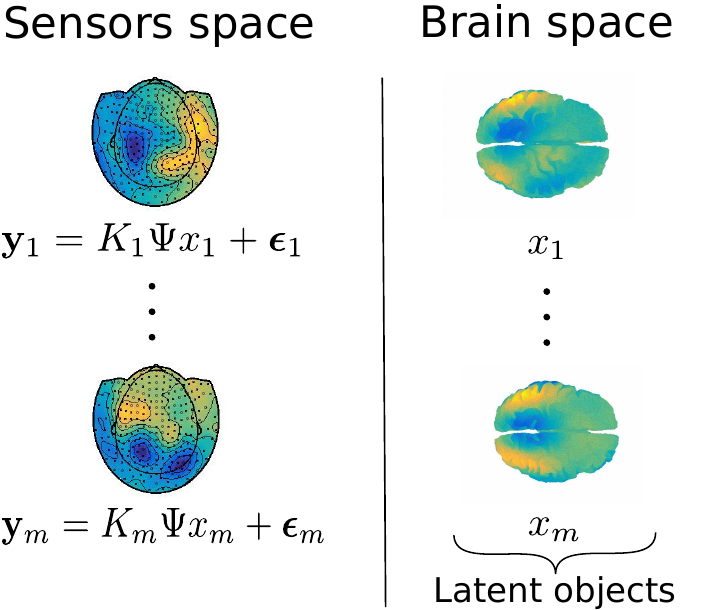}
\caption[]{Illustration of the setting introduced with model (\ref{eq:model}). }
\label{fig:illustration_func}
\end{figure}

We consider the problem of estimating the PC functions $\{\psi_r\}$ in (\ref{eq:model}), and associated scores $\{\zeta_{l,r}\}$, from the observations $\{\vect{y}_l\}$. In Figure~\ref{fig:illustration_func} we give an illustration of the introduced setting. Note that it would not be necessary to define the evaluation operator if the forward operators were defined to be functionals $\{\calK_l: \calF \ra \R^{p}\}$, relating directly the functional objects on the brain space to the real vectors on the sensors space. It is however the case that the operators $\{K_l\}$ are computed in a matrix form by third party software (see Section~\ref{sec:application} for details) for a pre-specified set of points $\{v_j\} \subset \M$ and it is thus convenient to take this into account in the model through the introduction of an evaluation operator $\Psi$.

In the case of single subject studies, the surface $\M$ is the subject's reconstructed cortical surface, an example of which is shown on the right panel of Figure~\ref{fig:forwardoperator}. In this case, it is natural to assume that there is one common forward operator $K$ for all the detected signals. In the more general case of multi-subject studies, $\M$ is assumed to be a template cortical surface. We are thus assuming that the individual cortical surfaces have been registered to the template $\M$, which means that there is a smooth and one-to-one correspondence between the points on each individual brain surface and the template surface $\M$, where the PC functions are defined.

However, notice that when it comes to the computation of the forward operators, we are not assuming the brain geometries of the single subjects to be all equal to a geometric template, as in fact the model in (\ref{eq:model}) allows for sample-specific forward operators $\{K_l\}$. The individual cortical surfaces could also have different number of mesh points, in that case the subject-specific `resampling' operator could be absorbed into the definition of sample-specific evaluation operators $\{\Psi_l\}$.

The estimation of the PC functions in (\ref{eq:model}) has been classically dealt with by reconstructing each observation $x_l$ independently and subsequently performing PCA. However, such an approach can be sub-optimal in particular in a low signal-to-noise setting, as when estimating one signal, the information from all the other sampled signals is systematically ignored. The statistical analysis of data samples that are random functions or surfaces, i.e. functional data, has also been explored in the Functional Data Analysis (FDA) literature \citep{ramsay2005}, however, most of those works focus on the setting of fully observed functions. An exception to this is the sparse FDA literature \citep[see e.g.][]{Yao2005}, where instead the functional samples are assumed to be observable only through irregular and noisy evaluations.

In the case of direct but noisy observations of a signal, previous works on statistical estimation of the covariance function, and associated eigenfunctions, have been made, for instance, in \cite{bunea2015} for regularly sampled functions and in \cite{Yao2005} and \cite{huang2008} for sparsely sampled functions. A generalization to functions whose domain is a manifold is proposed in \cite{Lila2016} and appropriate spatial coherence is introduced by penalizing directly the eigenfunctions of the covariance operator to be estimated, i.e. the PC functions. In the indirect observations setting, \cite{Tian2012} propose a separable model in time and space for source localization. The estimation of PC functions of functional data in a linear space and on linear domains, from indirect and noisy samples, has been previously covered in \cite{Amini2012}. They propose a regularized M-estimator in a Reproducing Kernel Hilbert Space (RKHS) framework. Due to the fact that in practice the introduction of a RKHS relies on the definition of a kernel, i.e. a covariance function on the domain, this approach cannot be easily extended to non-linear domains. In \cite{Katsevich2015}, driven by an application to cryo-electron microscopy, the authors propose an unregularized estimator for the covariance matrix of indirectly observed functions. However, a regularized approach is crucial in our setting, due to the strong ill-posedness of the inverse problem considered. In the discrete setting, also other forms of regularization have been adopted, e.g. sparsity on the inverse covariance matrix \citep{Friedman2008,Liu2019}.

\subsection{The case of indirectly observed covariance operators}\label{sec:ismcovpca_intro}

A natural generalization of the setting introduced in the previous section is considering observations that have group specific covariance operators. In detail, suppose now we are given a set of $n$ covariance functions $\{C_i: i=1,\ldots,n\}$, representing the underlying covariance operators $\{\calC_i: i=1,\ldots,n\}$ on the brain space. In our driving application, each covariance function $C_i:\M \times \M \ra \R$ describes the functional connectivity of the $i$th individual or the functional connectivity of the same individual at the $i$th time-point. 

\begin{figure}[!htb] 
\centering
\includegraphics[width=0.6\textwidth]{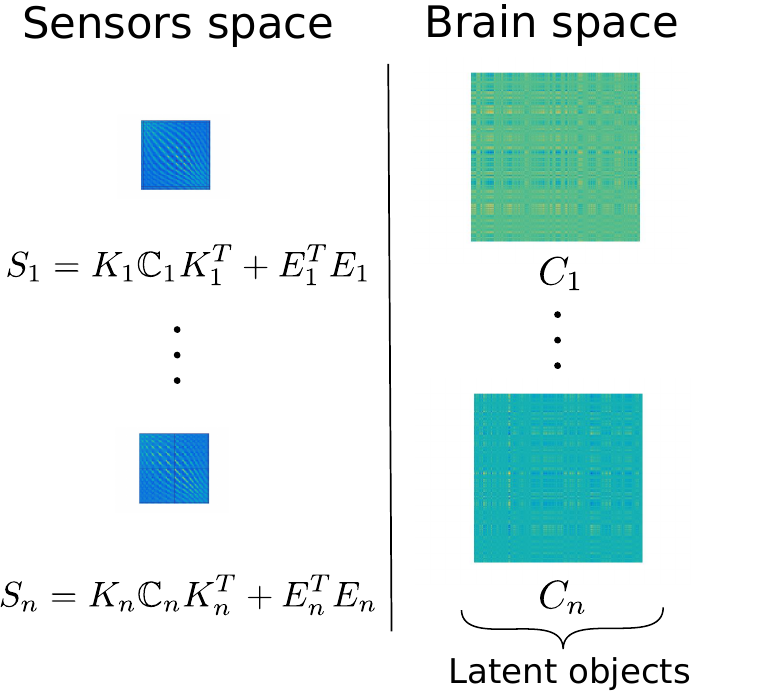}
\caption[]{Illustration of the setting introduced with model (\ref{eq:cov_brain_sensors}).}
\label{fig:illustration_cov}
\end{figure}

We consider the problem of defining and estimating a set of covariance functions, that we call PC covariance functions, which enable the description of $\{C_i\}$ through the `linear combinations' of few components. Such a reduced order description is of interest, for example, in understanding how functional connectivity varies among individuals or over time.

We define a model for the PC covariance functions of $\{C_i\}$ from the set of indirectly observed covariance matrices, computed from the signals on the sensors space, and thus given by $\{S_i\in \R^{s \times s}, i =1,\ldots,n \}$ with
\begin{equation}\label{eq:cov_brain_sensors}
S_i = K_i \mathbb{C}_i K_i^T + E_i^T E_i, \qquad i=1,\ldots,n,
\end{equation}
where $\mathbb{C}_i = (C_i(v_j, v_{j'}))_{jj'}$, and $\{v_j: j=1,\ldots,p\}$ are the sampling points associated with the operator $\Psi$.
The forward operators $\{K_i\}$ act on both sides of the covariance functions $\{C_i\}$, due to the linear transformation $K_i \Psi$ applied to the signals on the brain space before being detected on the sensors space. The term $E_i^T E_i$ is an error term, where $E_i$ is a $s \times s$ matrix such that each entry is an independent sample of a Gaussian distribution with mean zero and standard deviation $\sigma$. Model (\ref{eq:cov_brain_sensors}) could be regarded as an implementation of the idealized Problem~\ref{eq:cov_model_intro}, where the covariance operators are represented by the associated covariance functions. An illustration of the setting introduced can be found in Figure~\ref{fig:illustration_cov}.

The problem introduced in this section has not been extensively covered in the literature. In the discrete case, \cite{Dryden2009} introduce a tangent PCA model for directly observed covariance matrices. An extension to directly observed covariance operators has been proposed in \cite{Pigoli2014}. Also related to our work is the setting considered in \cite{Petersen2019}, where the authors propose a regression framework for responses that are random objects (e.g. covariance matrices) with Euclidean predictors. The proposed regression model is applied to study associations between age and low-dimensional correlation matrices, representing functional connectivity, which have been computed from a parcellation of the brain. In Section~\ref{sec:PCcovariances}, we propose a novel PCA approach for indirectly observed high-dimensional covariance matrices.

\section{Principal components of indirectly observed functions}\label{sec:PCfunctions}

The aim of this section is to define a model for the estimation of the PC functions $\{\psi_r\}$ from the observations $\{\vect{y}_l\}$, defined in (\ref{eq:model}). Although the model proposed in this section is not the main contribution of this work, it allows us to introduce some of the concepts necessary to the definition of the more general model for indirectly observed covariance functions in Section~\ref{sec:PCcovariances}.

\subsection{Model}
Let now $\vect{z} = (z_1,\ldots,z_m)^T$ be a $m$-dimensional
real column vector and $H^2(\M)$ be the Sobolev space of functions in $L^2(\M)$ with first and second distributional derivatives in $L^2(\M)$. From now on $\calF$ is instantiated with $H^2(\M)$. We propose to estimate  $\hat{f} \in H^2(\M)$, the first PC function of $X$, and the associated PC scores vector $\vect{z}$, by solving the equation
\begin{equation}\label{eq:model_f_pca}
(\hat{\vect{z}},\hat{f}) = \argmin \limits_{\vect{z} \in \R^m,f \in H^2(\M)} \sum \limits_{l=1}^m \| \vect{y}_l - z_l K_l \Psi f \|^2 + \lambda \vect{z}^T\vect{z} \int_{\mathcal{M}} \! \Delta^2_{\mathcal{M}} f,
\end{equation}
where $\| \cdot \|$ is the Euclidean norm and $\Delta$ is the Laplace-Beltrami operator, which enables a smoothing regularizing effect on the PC function $\hat{f}$. The data fit term encourages $K_l \Psi f$ to capture the strongest mode of variation of $\{\vect{y}_l\}$. The parameter $\lambda$ controls the trade-off between the data fit term of the objective function and the regularizing term. The second PC function can be estimated by classical deflation methods, i.e. by applying model (\ref{eq:model_f_pca}) on the residuals $\{\vect{y}_l - \hat{z}_l K_l \Psi \hat{f}\}$, and so on for the subsequent PCs. The proposed model can be interpreted as a regularized least square estimation of the first PC function $\psi_1$ in (\ref{eq:model}), with the terms $\{z_l\}$ playing the role of estimates of the variables $\{\zeta_{l,1}\}$.

In the simplified case of a single forward operator $K := K_1 = \ldots = K_m$, the minimization problem (\ref{eq:model_f_pca}) can be reformulated in a more classical form. In fact, fixing $f$ in (\ref{eq:model_f_pca}) and minimizing over $\vect{z}$ gives
\begin{equation}\label{eq:reformulation_eigenvalues}
z_l = \frac{\vect{y}_l^T K \Psi f}{\| K \Psi f \|^2 + \lambda \int_{\mathcal{M}} \! \Delta^2_{\mathcal{M}} f }, \qquad l=1,\ldots,m,
\end{equation}
which can then be used to show that the minimization problem (\ref{eq:model_f_pca}) is equivalent to maximizing
\begin{equation}\label{eq:reformulation_eigenvector}
\frac{ (\Psi f)^T K^T \mathbb{Y}^T \mathbb{Y} K (\Psi f)}{\| K \Psi f\|^2 + \lambda \int_{\mathcal{M}} \! \Delta^2_{\mathcal{M}} f},
\end{equation}
with $\mathbb{Y}$ a $m \times s$ real matrix, where the $l$th row of $\mathbb{Y}$ is the observation $\vect{y}^T_l$. This reformulation gives further insights on the interpretation of $\hat{f}$ in (\ref{eq:model_f_pca}). In fact, $\hat{f}$ is such that $K \Psi \hat{f}$ maximizes $(K \Psi \hat{f})^T \{\frac{1}{m} \mathbb{Y}^T \mathbb{Y} \}(K \Psi \hat{f})$ subject to a norm constraint. The term $\{\frac{1}{m} \mathbb{Y}^T \mathbb{Y} \}$ is the empirical covariance matrix in the sensors space. The term $\vect{z}^T\vect{z}$ in (\ref{eq:model_f_pca}) places the regularization term $\lambda \int_{\mathcal{M}} \! \Delta^2_{\mathcal{M}} f$ in the  denominator of the equivalent formulation (\ref{eq:reformulation_eigenvector}). Thus, $\hat{f}$ is regularized by the choice of norm in the denominator of (\ref{eq:reformulation_eigenvector}), in a similar fashion to the classic functional principal component formulation of \cite{silverman1996}. Ignoring the spatial regularization, the point-wise evaluation of the PC function $\Psi f$ in (\ref{eq:reformulation_eigenvector}) can be interpreted as the first PC vector computed from the dataset of backprojected data $[K_1^T \vect{y}_1, \ldots, K^T_m\vect{y}_m]^T$, similarly to what is proposed in \cite{Dobriban2017} in the context of optimal prediction.

\subsection{Algorithm}\label{sec:Algorithm_PCfunctions}
Here we propose a minimization approach for the objective function in (\ref{eq:model_f_pca}), which we approach by alternating the minimization of $\vect{z}$ and $f$ in an iterative algorithm. In (\ref{eq:model_f_pca}), a normalization constraint must be considered to make the representation unique, as in fact multiplying $\vect{z}$ by a constant and dividing $f$ by the same constant does not change the objective function. We optimize in $\vect{z}$ under the constraint $\|\vect{z}\| = 1$, which leads to a normalized version of the estimator (\ref{eq:reformulation_eigenvalues})
\begin{equation}\label{eq:reformulation_eigenvalues_normalized}
z_l = \frac{\vect{y}_l^T K_l \Psi f}{\sqrt{ \sum_{l=1}^{m}\vect{y}_l^T K_l \Psi f }}, \qquad l=1,\ldots,m.
\end{equation}

For a given $\vect{z}$, solving (\ref{eq:model_f_pca}) with respect to $f$ will turn out to be equivalent to solving an inverse problem, which we discretize adopting a Mixed Finite Elements approach \citep{Azzimonti2014}. Specifically, consider now a triangulated surface $\mathcal{M}_\mathcal{T}$, union of the finite set of triangles $\mathcal{T}$, giving an approximated representation of the manifold $\mathcal{M}$. We then consider the linear finite element space $V$ consisting of a set of globally continuous functions over $\mathcal{M}_\mathcal{T}$ that are affine where restricted to any triangle $\tau$ in $\mathcal{T}$, i.e.
\begin{equation*}
V = \{ v \in C^0(\mathcal{M}_\mathcal{T}): v|_{\tau} \text{ is affine for each } \tau \in \mathcal{T} \}.
\end{equation*}

This space is spanned by the nodal basis $\phi_1, \ldots, \phi_\kappa$  associated with the nodes $\xi_1,\ldots,\xi_\kappa$, corresponding to the vertices of the triangulation $\mathcal{M}_\mathcal{T}$. Such basis functions are Lagrangian, meaning that $\phi_i(\xi_j)=1$ if $i=j$ and $\phi_i(\xi_j)=0$ otherwise. Setting $\vect{c} = (f(\xi_1),\ldots,f(\xi_\kappa))^T$ and $\pmb{\phi} = (\phi_1,\ldots,\phi_\kappa)^T$, every function $f \in V$ has the form
\begin{equation}\label{eq:basis}
f(v) = \sum_{k=1}^\kappa f(\xi_k) \phi_k(v) = \vect{c}^T \pmb{\phi}(v)
\end{equation}
for all $v \in \mathcal{M}_\mathcal{T}$. To ease the notation, we assume that the $p$ points $\{v_j\}$ associated with the evaluation operator $\Psi$ coincide with the nodes of the triangular mesh $\xi_1,\ldots,\xi_\kappa$, and thus we have that the coefficients $\vect{c}$ are such that $\vect{c} = \Psi f$ for any $f \in V$. Consequently, we are assuming the forward operators $\{K_l\}$ to be $s \times \kappa$ matrices, relating the $\kappa$ points on the cortical surface of the $i$th sample, in one-to-one correspondence to $\xi_1,\ldots,\xi_\kappa$, to the $s$-dimensional signal detected on the sensors for the $i$th sample.

Let now $M$ and $A$ be the mass and stiffness $\kappa \times \kappa$ matrices defined as $(M)_{jj'} =\int_{\mathcal{M}_{\mathcal{T}}} \phi_j \phi_{j'}$ and $(A)_{jj'}=\int_{\mathcal{M}_\mathcal{T}} \nabla_{\mathcal{M}_\mathcal{T}} \phi_j \cdot \nabla_{\mathcal{M}_\mathcal{T}} \phi_{j'}$, where $\nabla_{\mathcal{M}}$ is the gradient operator on the manifold $\M$. Practically, $\nabla_{\mathcal{M}_\mathcal{T}} \phi_{j}$ is a constant function on each triangle of $\mathcal{M}_\mathcal{T}$, and can take an arbitrary value on the edges\footnote{Formally, these are weak derivatives, hence uniquely defined almost everywhere (i.e. up to a set of measure zero) and are always evaluated in an integral form \citep[see][for further details]{Dziuk2013}.}.

Let $h = \max_{\tau \in \calT}(diam(\tau))$ denote the maximum diameter of the triangles forming $\M_\calT$, then the solution $\hat{f}_h$ of (\ref{eq:model_f_pca}), in the discrete space $V$, is given by the following proposition.

\begin{proposition}\label{prop:FE_func}
The Surface Finite Element solution $\hat{f}_h \in V$ of model (\ref{eq:model_f_pca}), for a given unitary norm vector $\vect{z}$, is $\hat{f}_h = \vect{\hat{c}}^T \pmb{\phi}$ where $\vect{\hat{c}}$ is the solution of
\begin{equation}\label{eq:model_f_pca_discrete}
\vect{\hat{c}} = (\sum_{l=1}^m z_l^2 K_l^TK_l + \lambda A M^{-1} A)^{-1} \sum_{l=1}^m z_l K_l^T \vect{y}_l.
\end{equation}
\end{proposition}

Equation (\ref{eq:model_f_pca_discrete}) has the form of a penalized regression, where the discretized version of the penalty term is $A M^{-1} A$.

\begin{algorithm}[H]
\caption{Inverse Problems - PCA Algorithm}\label{alg:InversefPCA}
\begin{algorithmic}[1]
\State Initialization:
\begin{enumerate}[label=(\alph*)]
\item Computation of $M$ and $A$
\item Initialize $\vect{z}$, the scores vector associated with the first PC function
\end{enumerate}
\State PC function's estimation:

Compute $\vect{c}$ such that %
\begin{equation*}
\bigg(\sum_{l=1}^m z_l^2 K_l^TK_l + \lambda A M^{-1} A \bigg) \vect{c} = \sum_{l=1}^m z_l K_l^T \vect{y}_l
\end{equation*}
\begin{equation*}
f_h \la \vect{{c}}^T \pmb{\phi}
\end{equation*}
\State Scores estimation:
\begin{equation*}
z_l \la \frac{\vect{y}_l^T K_l \Psi f_h}{\sqrt{ \sum_{l=1}^{m}\vect{y}_l^T K_l \Psi f_h }}, \qquad l = 1,\ldots,m
\end{equation*}

\State Repeat Steps 2--3 until convergence
\end{algorithmic}
\end{algorithm}
The sparsity of the linear system (\ref{eq:model_f_pca_discrete}), namely the number of zeros, depends on the sparsity of its components. The matrices $M$ and $A$ are very sparse, however $M^{-1}$ is not, in general. To overcome this problem, in the numerical analysis of Partial Differential Equations literature, the matrix $M^{-1}$ is generally replaced with the sparse matrix $\tilde{M}^{-1}$, where $\tilde{M}$ is the diagonal matrix such that $\tilde{M}_{jj} = \sum_{j'} M_{jj'}$ \citep{Fried1975,Zienkiewicz2013}. The penalty operator $A \tilde{M}^{-1} A$ approximates very well the behavior of $A M^{-1} A$.

Moreover, in the case of longitudinal studies that involve only one subject, we have a single forward operator $K := K_1 = \ldots = K_m$ common to all the observed signals, and consequently equation (\ref{eq:model_f_pca_discrete}) can be rewritten as the sparse overdetermined system
\begin{equation}\label{eq:model_f_pca_ls}
	\begin{bmatrix}
		K&\\
		\sqrt{\lambda}{\tilde{M}}^{-\nicefrac{1}{2}} A
	\end{bmatrix}
		\vect{c}
=
	\begin{bmatrix}
		Y^T\vect{z}\\
		\vect{0}
	\end{bmatrix},
\end{equation}
to be interpreted in a least-square sense. A sparse QR solver can be finally applied to efficiently solve the linear system (\ref{eq:model_f_pca_ls}).

In Algorithm~\ref{alg:InversefPCA} we summarize the main algorithmic steps to compute the PC functions and associated PC scores for indirectly observed functions. The initializing scores $\vect z$ can be chosen either at random or, when there is a correspondence between the detectors of different samples (e.g. $K_1 = \ldots = K_m$), with the scores obtained by performing PCA on the observations in the sensors space.

\subsection{Eigenfunctions of indirectly observed covariance operators}\label{sec:eig_cov_op}
Suppose now we are in the case of a single forward operator $K$. Combining Steps 2--3 of Algorithm~\ref{alg:InversefPCA}, and moving the normalization step from $(z_l)$ to $f_h$, we obtain the iterations
\begin{align*}
&(K^T K + \lambda A M^{-1} A ) \vect{c} = K^T \sum_{l=1}^m (\vect{y}_l \vect{y}^T_l) K \Psi f_h\\
&f_h \la \vect{{c}}^T \pmb{\phi}; \,f_h \ra \frac{f_h}{\|f_h\|}.
\end{align*}
The obtained algorithm depends on the data only through $\sum_{l=1}^m \vect{y}_l \vect{y}^T_l$ that up to a constant is the covariance matrix computed on the sensors space. The proposed algorithm can thus be applied to situations where the observations $\{\vect{y}_l\}$ are not available, but we are given only the associated $s \times s$ covariance matrix on the sensors space, computed from $\{\vect{y}_l\}$. This could be of interest in situations where the temporal resolution is very high and the spatial resolution is low, therefore it is convenient to store the covariance matrix rather than the entire set of observations.

\section{Reconstruction and representation of indirectly observed covariance operators}\label{sec:PCcovariances}
Consider now $n$ sample covariance matrices $S_1,\ldots,S_n$, each of size $s \times s$, representing $n$ different connectivity maps on the sensors space. Three of such covariance matrices, associated with three different individuals, are shown in Figure~\ref{fig:covariances}. Recall moreover that we denote with $\M$ the brain surface template and with $\{K_i \in \R^{s \times p}\}$ the set of subject-specific forward operators, relating the signal at the $p$ pre-specified points $\{v_j\}$ on the cortical surface $\M$ with the signal detected on the $s$ sensors.

The aim of this section is to introduce a model for the reconstruction and representation of the covariance functions $\{C_i\}$, on the brain space, associated with the actually observed covariance matrices $\{S_i\}$, on the sensors space. The matrices $\{S_i\}$ are related to the covariance functions $\{C_i\}$ through formula (\ref{eq:cov_brain_sensors}) that we recall here being
\begin{equation*} 
S_i = K_i \mathbb{C}_i K_i^T + E_i^T E_i, \qquad i=1,\ldots,n,
\end{equation*}
with $\mathbb{C}_i = (C_i(v_j, v_l))_{jl}$, and $\{v_j\}$ the sampling points associated with the operator $\Psi$.

First, in Section~\ref{sec:PCcov_subj}, we see how the PC model introduced in Section~\ref{sec:PCfunctions} could be applied to individually reconstruct the covariance functions $\{C_i\}$. In Section~\ref{sec:PCcov_pop}, we introduce a population model that achieves both reconstruction and joint representation of the underlying covariance functions $\{C_i\}$.

\subsection{A subject-specific model}\label{sec:PCcov_subj}
Let $S_i^{\nicefrac{1}{2}}$ be a square-root decomposition of $S_i$, which is a decomposition such that $S_i = (S_i^{\nicefrac{1}{2}})^T S_i^{\nicefrac{1}{2}}$, for all $i=1,\ldots,n$. This could be given, for instance, by $S_i^{\nicefrac{1}{2}} = D_i^{\nicefrac{1}{2}} V_i^T$ where $S_i = V_iD_iV_i^T$ is the spectral decomposition of $S_i$ and $D_i^{\nicefrac{1}{2}}$ denotes the diagonal matrix whose entries are the square-root of the (non-negative) entries of $D_i$. Each square-root decomposition $S_i^{\nicefrac{1}{2}}$ can be interpreted as a data-matrix whose empirical covariance is $S_i$. Another possible choice for the square-root decompositions is $S_i^{\nicefrac{1}{2}} = V_i D_i^{\nicefrac{1}{2}} V_i^T$. The output of the proposed algorithms will not depend on the specific choice of the square-root decompositions.

In the most general setting, each covariance matrix $S_i$ is an indirect observation of an underlying covariance function $C_i$, which can be expressed in terms of its spectral decomposition as
\[
C_i(v,v') = \sum_{r=1}^\infty \gamma_{ir} \psi_{ir}(v)\psi_{ir}(v'), \qquad \forall v,v' \in \M,
\]
where, for each $i$, $\gamma_{i1} \geq \gamma_{i2} \geq \cdots \geq 0$ is a sequence of non-increasing variances and $\{\psi_{ir}\}_r$ a set of orthonormal eigenfunctions.
Introduce now $\{\hat{f}_{i} \in H^2(\M)\}$ and $\{\vect{\hat{z}}_i \in \R^s\}$, obtained by applying model (\ref{eq:model_f_pca}) to each sample independently, i.e.
\begin{equation}\label{eq:model_cov_pca}
\left\{(\hat{\vect{z}}_i,\hat{f}_i)\right\}_i = \argmin \limits_{\{\vect{z}_i\} \subset \R^s,\{f_i\} \subset H^2(\M)} \|S_i^{\nicefrac{1}{2}} - \vect{z}_i (K_i \Psi f_i)^T \|^2_F + \lambda \| \vect{z}_i\|^2 \int_{\mathcal{M}} \! \Delta^2_{\mathcal{M}} f_i, \qquad i=1,\ldots,n,
\end{equation}
with $\| \cdot \|_F$ denoting the Frobenius matrix norm. Each estimate $\hat{f}_i$, from model (\ref{eq:model_cov_pca}), can be interpreted as a regularized estimate of the leading PC function of $S_i^{\nicefrac{1}{2}}$ and thus of the eigenfunction $\psi_{i1}$. The subsequent eigenfunctions can be estimated by deflation methods, i.e. by removing the estimated components $\vect{\hat z}_i (K_i \Psi \hat f_i)^T$ from $S_i^{\nicefrac{1}{2}}$ and reapplying model (\ref{eq:model_cov_pca}). This leads to a set of estimates $\{\hat{f}_{ir}\}$ and $\{\hat{\vect{z}}_{ir}\}$. 

The unregularized version of model (\ref{eq:model_cov_pca}) is equivalent to a Singular Value Decomposition applied to each matrix $S_i^{\nicefrac{1}{2}}$ independently, which would lead to a set of orthogonal estimates $\{\hat{\vect{z}}_{ir}\}_r \subset \R^s$, for each $i=1,\ldots,n$. In the regularized model orthogonality is not enforced, however the estimated PC components can be orthogonalized post-estimation by means of a QR decomposition.

Define now the empirical variances to be $\hat{\gamma}_{ir} = \|\hat f_{ir}\|^2_{L^2(\M)}$ and consider the $L^2(\M)$-normalized version of $\{\hat f_{ir}\}$. An approximate representation of $S_i = (S_i^{\nicefrac{1}{2}})^T S_i^{\nicefrac{1}{2}}$ is thus given by
\begin{equation}\label{eq:repr_cov_pca}
S_i =  K_i \sum_r \left\{ \hat{\gamma}_{ir} (\Psi \hat{f}_{ir}) (\Psi \hat{f}_{ir})^T \right\} K_i^T,
\end{equation}
and the associated approximate representation of $C_i$, in terms of $\{\hat{\gamma}_{ir}\}$ and $\{\hat{f}_{ir}\}$, is
\[
C_i = \sum_r \hat{\gamma}_{ir} \hat{f}_{ir} \otimes \hat{f}_{ir},
\]
where $\hat \gamma_{ir}$ is an estimate of the variance $\gamma_{ir}$ and $\hat{f}_{ir}$ is an estimate of $\psi_{ir}$. The tensor product $\hat{f}_{ir} \otimes \hat{f}_{ir}$ is such that $(\hat{f}_{ir} \otimes \hat{f}_{ir})(v,v') = \hat{f}_{ir}(v) \hat{f}_{ir}(v')$ for all $v,v' \in \M$. The regularizing terms in (\ref{eq:model_cov_pca}) introduce spatial coherence on the estimated $\{\hat{f}_{ir}\}$ and thus on the estimated eigenfunctions of $\{C_i\}$, fundamental in an inverse problems setting. 

The reconstructed covariance functions $\{C_i\}$ could be discretized on a dense grid, leading to a collection of covariance matrices $(C_i(v_j, v_l))_{jl}$. Following the approach in \cite{Dryden2009}, a Riemannian metric could be defined on the space of covariance matrices, followed by projection of $(C_i(v_j, v_l))_{jl}$ on the tangent space centered at the sample Fr\'echet mean. PCA could then be carried out on vectorizations of the tangent space representations. A related approach, for covariance functions, has been adopted in \cite{Pigoli2014}. 

However, the aforementioned approaches could be prohibitive in our setting. In fact, performing PCA on tangent space projections produces modes of variation that are geodesics passing through the mean, and whose interpretation in a high-dimensional setting is often challenging. Therefore, in the next section, we propose an alternative model that enables joint reconstruction, and representation on a `common basis', of indirectly observed covariance functions.

\subsection{A population model}\label{sec:PCcov_pop}
Let $\{\vect{\hat{z}}_i\}_{i=1}^n \subset \R^s$ and $\hat{f} \in H^2(\M)$ be given by the following model:
\begin{equation}\label{eq:model_cov_pca_comm}
(\{\hat{\vect{z}}_i\},\hat{f}) = \argmin \limits_{\{\vect{z}_i\} \subset \R^s,f \in H^2(\M)} \sum \limits_{i=1}^n \|S_i^{\nicefrac{1}{2}} - \vect{z}_i (K_i \Psi f)^T \|^2_F + \lambda \sum \limits_{i=1}^n \| \vect{z}_i\|^2 \int_{\mathcal{M}} \! \Delta^2_{\mathcal{M}} f.
\end{equation}
The newly defined model, as opposed to model (\ref{eq:model_cov_pca}), has now a subject-specific $s$-dimensional vector $\vect{z}_i$ and a term $f$ that is common to all samples. As in the previous model, the subsequent components can be estimated by deflation methods, leading to a set of estimates $\hat{f}_{r}$ and $\hat{\vect{z}}_{ir}$. 

Define now the empirical variances to be $\hat{\gamma}_{ir} = \|\vect{\hat{z}}_{ir}\|^2 \|\hat f_{r}\|^2_{L^2(\M)}$ and consider the $L^2(\M)$-normalized version of $\{\hat f_{r}\}$. The empirical term in model (\ref{eq:model_cov_pca_comm}) suggests an approximate representation of $S_i$ that is 
\begin{equation}\label{eq:repr_cov_pca_comm}
C_i = \sum_r \hat{\gamma}_{ir} \hat{f}_{r} \otimes \hat{f}_{r},
\end{equation}
where each underlying covariance function $C_i$ is approximated by the sum of the product between a subject-specific constant $\hat{\gamma}_{ir}$ and a component $\hat{f}_{r} \otimes \hat{f}_{r}$ common to all the observations. The regularizing term in (\ref{eq:model_cov_pca_comm}) introduces spatial coherence on the estimated functions $\{\hat{f}_{r}\}$.

The covariance operators $\{\calC_i\}$ are said to be commuting if $\calC_i\calC_{i'} = \calC_{i'}\calC_i$ for all $i,i' = 1, \ldots, n$. This property can be equivalently characterized as
\begin{equation}\label{eq:cov_comm_basis}
C_i(v,v') = \sum_{r=1}^\infty \gamma_{ir} \psi_{r}(v)\psi_{r}(v'), \qquad \forall v,v' \in \M,
\end{equation}
with $\{\gamma_{ir}\}_r$ subject-specific variances and $\{\psi_{r}\}$ a set of common orthonormal functions. Thus, a collection of commuting covariance operators is such that its covariance operators can be simultaneously diagonalized by a basis $\{ \psi_{r}\}$. In this case, the functions $\{\hat{f}_r\}$ can be regarded as estimates of $\{\psi_{r}\}$ and $\{\hat \gamma_{ir}\}$ estimates of $\{\gamma_{ir}\}$. 

On the one hand, model (\ref{eq:model_cov_pca_comm}) constrains the estimated covariances to be of the form $C_i = \sum_r \hat \gamma_{ir} \hat{f}_{r} \otimes \hat{f}_{r}$ and not of the more general form $C_i = \sum_r \hat \gamma_{ir} \hat{f}_{ir} \otimes \hat{f}_{ir}$. On the other hand, such a model takes advantage of all the $n$ samples to estimate the components $\{\hat{f}_{r} \otimes \hat{f}_{r}\}$. Moreover, the associated variables $\{\hat \gamma_{ir}\}$ give a convenient approximate description of the $i$th covariance, as they are comparable across samples, as opposed to the one computed from model (\ref{eq:model_cov_pca}). In fact, the $i$th covariance function can be represented by the variance vector $(\hat \gamma_{i1},\ldots,\hat \gamma_{iR})^T$, for a suitable truncation level $R$, where each entry is associated with the rank-one component $\hat{f}_{r} \otimes \hat{f}_{r}$. For each $r$, a scatter plot of the variances $\{\gamma_{ir}\}_i$, as the one in Figure~\ref{fig:scores_covfunctions_multi}, helps understand what the average contribution of the $r$th components is and what its variability across samples is. Model (\ref{eq:repr_cov_pca_comm}) could also be interpreted as a common PCA model \citep{Flury1984,Benko2009}, as $\{\hat f_r\}$ are the estimated regularized eigenfunctions of the pooled covariance $C = \frac{1}{n} \sum_{i=1}^{n} C_i$.

Potentially, PCA could be performed on the descriptors $(\hat \gamma_{i1},\ldots,\hat \gamma_{iR})^T$ to find rank-$R$ components that maximize the variance of linear combinations of $\{\hat{\gamma}_{ir}\}$ (i.e. the variance of the variances). However, results would be more difficult to interpret, as they would involve variations that are rank-$R$ covariance functions around the rank-$R$ mean covariance function.

\subsection{Algorithm}
The minimization in (\ref{eq:model_cov_pca}), for each fixed $i$, is a particular case of the one in (\ref{eq:model_f_pca}) (see Section~\ref{sec:Algorithm_PCfunctions}), so we focus on the minimization problem in (\ref{eq:model_cov_pca_comm}) which is also approached in an iterative fashion. We set $\sum_{i=1}^n \| \vect{z}_i\|^2 = 1$ in the estimation procedure. This leads to the estimates of $\{\vect{z}_i\}$, given $f$, that are
\[
\vect{z}_i = \frac{\tilde{\vect{z}}_i}{\sqrt{\sum_{i=1}^n \| \tilde{\vect{z}}_i \|^2}}, \qquad i = 1,\ldots,n,
\]
with
\begin{equation*}
\tilde{\vect{z}}_i = S_i^{\nicefrac{1}{2}} K_i \Psi f_h, \qquad i = 1,\ldots,n.
\end{equation*}

The estimate of $f$ given $\{\vect{z}_i\}$, in the discrete space $V$ introduced in Section~\ref{sec:Algorithm_PCfunctions}, is given by the following proposition.
\begin{proposition}\label{prop:FE_cov}
The Surface Finite Element solution $\hat{f}_h \in V$ of model (\ref{eq:model_cov_pca_comm}), given the vectors $\{\vect{z}_i\}$, is $\hat{f}_h = \vect{\hat{c}}^T \pmb{\phi}$ where $\vect{\hat{c}}$ is the solution of
\begin{equation}
\bigg(\sum_{i=1}^n \|\vect{z}_i \|^2 K_i^TK_i + \lambda A M^{-1} A \bigg) \vect{\hat{c}} =  \sum_{i=1}^n K_i^T (S_i^{\nicefrac{1}{2}})^T\vect{z}_i.
\end{equation}
\end{proposition}%

\begin{algorithm}[!htb]
\caption{Inverse Problems - Covariance PCA Algorithm}\label{alg:InverseCovfPCA}
\begin{algorithmic}[1]
\State Square-root decompositions
\begin{enumerate}[label=(\alph*)]
\item Compute the representations $S^{\nicefrac{1}{2}}_1,\ldots,S^{\nicefrac{1}{2}}_n$ from $S_1,\ldots,S_n$ as
\[
S_i^{\nicefrac{1}{2}} = D_i^{\nicefrac{1}{2}} V_i^T,
\]
with $S_i = V_iD_iV_i^T$ the spectral decomposition of $S_i$.
\end{enumerate}

\State Initialization:

\begin{enumerate}[label=(\alph*)]
\item Computation of $M$ and $A$
\item Initialize $\{\vect{z}_i\}_{i=1}^n$, the scores of the first PC
\end{enumerate}

\State PC function's estimation from model (\ref{eq:model_cov_pca}):

Compute $\vect{c}$ such that %
\begin{equation*}	
\bigg(\sum_{i=1}^n \|\vect{z}_i \|^2 K_i^TK_i + \lambda A M^{-1} A \bigg) \vect{{c}} = \sum_{i=1}^n K_i^T (S_i^{\nicefrac{1}{2}})^T\vect{z}_i
\end{equation*}
\begin{equation*}
f_h \la \vect{c}^T \pmb{\phi}
\end{equation*}
\State Scores estimation from model (\ref{eq:model_cov_pca}):
\begin{equation*}
\vect{z}_i \la S_i^{\nicefrac{1}{2}}  K_i \Psi f_h , \qquad i = 1,\ldots,n
\end{equation*}

\begin{equation*}
\vect{z}_i \la \frac{\vect{z}_i}{\sqrt{\sum_{i=1}^n \| \vect{z}_i \|^2}}, \qquad i = 1,\ldots,n
\end{equation*}

\State Repeat Step 3-4 until convergence

\end{algorithmic}
\end{algorithm}

Algorithm~\ref{alg:InverseCovfPCA} contains a summary of the estimation procedure. From a practical point of view, the choice to define the representation basis to be a collection of rank one (i.e. separable) covariance functions, of the type $F_r = \hat{f}_r \otimes \hat{f}_r$, is mainly driven by the following reasons. Firstly, rank-one covariance functions are easier to interpret due to their limited degrees of freedom. Secondly, on a rank one covariance function $F_r = \hat{f}_r \otimes \hat{f}_r$ spatial coherence can be imposed by regularizing $\hat{f}_r$, as in fact done for model (\ref{eq:model_cov_pca}), and this is fundamental in the setting of indirectly observed covariance functions. Finally, due to their size, it might not be possible to store the full reconstructions of the covariance functions $\{C_i\}$ on the brain space, instead, the representation model in (\ref{eq:repr_cov_pca_comm}) allows for an efficient joint representation of such covariance functions in terms of rank-one components.

\section{Simulations}\label{sec:simulations}
In this section, we perform simulations to assess the performances of the proposed algorithms. To reproduce as closely as possible the application setting, the cortical surfaces and the forward operators are taken from the MEG application described in Section~\ref{sec:application}. The details on the extraction and computation of such objects are left to the same section. For the same reason, the signals on the brain space considered here are vector-valued functions, specifically functions from the brain space $\M$ to $\R^3$, as is the case in the MEG application. The proposed methodology can be trivially extended to successfully deal with this case, as shown in the following simulations.

\subsection{Indirectly observed functions}
We consider $\M_\mathcal{T}$ to be a triangular mesh, with 8K nodes, representing the cortical surface geometry of a subject, as shown on the left panel of Figure~\ref{fig:forwardoperator}. Each of the 8K nodes will represent a location $v_j$ associated with the sampling operator $\Psi$. The locations of the nodes $\{v_j\}$ on the brain space, the location of the $241$ detectors on the sensors space and a model of the subject's head, enable the computation of a forward operator $K$ describing the relation between the signal generated on the locations $\{v_j\}$, on the brain space, and the signal detected on the $241$ sensors in the sensors space. In practice, the signal on each node $v_j$ is described by a three dimensional vector, characterized by an intensity and a direction, while the signal detected on the sensors space is a scalar signal. Thus, the forward operator is a $241 \times 24K$ matrix.

\begin{figure}[!htb]
\centering
\includegraphics[width=0.75\textwidth]{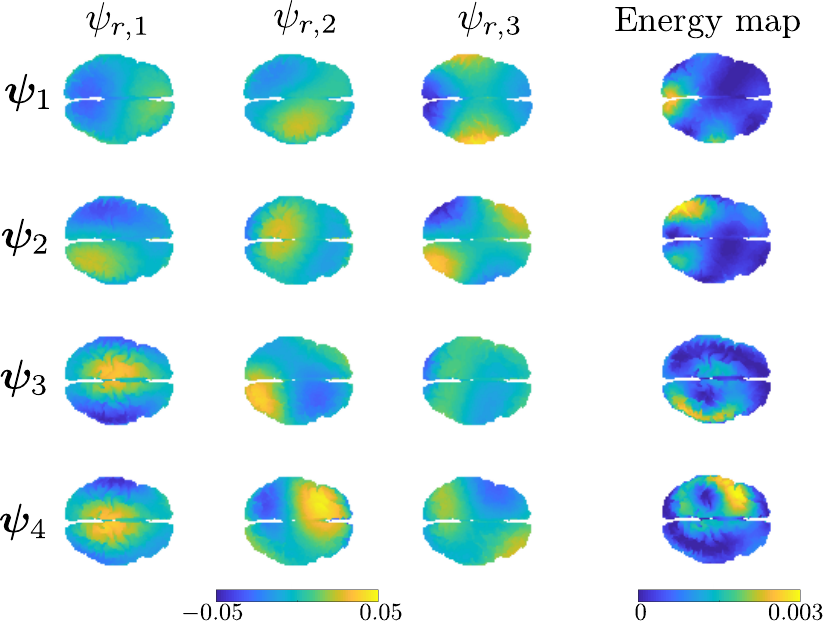}
\caption[]{From top to bottom the components and the energy maps of the PC functions $\bm{\psi}_1, \ldots, \bm{\psi}_4$.}
\label{fig:PCfunctions_sim}
\end{figure}

\begin{figure}[!htb]
\centering
\includegraphics[width=0.9\textwidth]{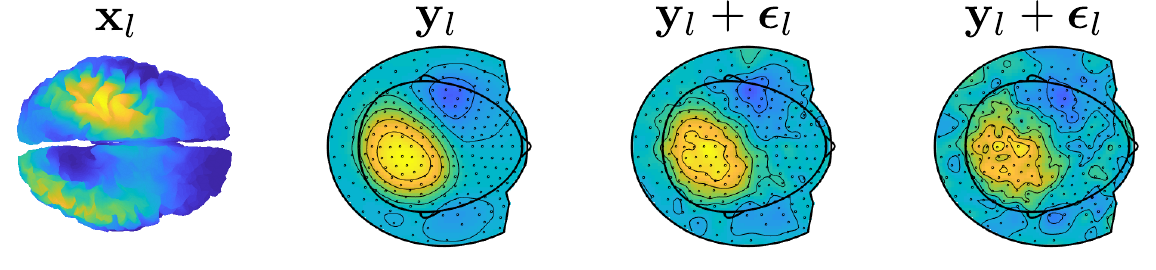}
\caption[]{From left to right, the energy map of a generated function $\vect{x}_l$, the associated signal $\vect{y}_l$ on the sensors space with respectively no additional error, Gaussian error of standard deviation $\sigma = 5$ and Gaussian error of standard deviation $\sigma = 10$.}
\label{fig:signal_sim}
\end{figure}

We first want to assess the performances of the proposed model in the case of indirect functional observations belonging to a linear space. To this purpose, we produce synthetic data following the generative model (\ref{eq:model}).
Specifically, on $\M_\mathcal{T}$, we construct the four $L^2(\M_\calT)$ orthonormal vector-valued functions $\{\bm{\psi}_r = (\psi_{r,1},\psi_{r,2},\psi_{r,3}): r=1,\ldots,4\}$, with $\bm{\psi}_r : \M_\mathcal{T} \ra \R^3$. These represent the PC functions to be estimated. In Figure~\ref{fig:PCfunctions_sim} we show the four components of $\{\bm{\psi}_r\}$ and the associated energy maps $\{\| \bm{\psi}_{r}(v)\|^2: v \in \M_\calT\}$, with $\| \cdot \|$ denoting the Euclidean norm in $\R^3$. We then generate $m = 50$ smooth vector-valued functions $\{\vect{x}_l\}$ on $\M_\mathcal{T}$ by
\[
\vect{x}_l = z_{l1} \bm{\psi}_1 + z_{l2} \bm{\psi}_2 + z_{l3} \bm{\psi}_3 + z_{l4} \bm{\psi}_4 \qquad l=1,\ldots,m,
\]
where $\{z_{lr}\}$ are i.i.d realizations of the four independent random variables $\{z_r \sim N(0, \gamma_r): r=1,\ldots,4\}$, with $\gamma_1 = 3^2$, $\gamma_2 = 2.5^2$, $\gamma_3 = 2^2$ and $\gamma_4 = 1$.

The functions $\{\vect{x}_l\}$ are sampled at the 8K nodes, and the forward operator is applied to the sampled values, producing a collection of vectors  $\{\vect{y}_l\}$ each of dimension $241$, the number of active sensors. Moreover, on each entry of the vectors $\{\vect{y}_l\}$, we add Gaussian noise with mean zero and standard deviation $\sigma$, for different choices of $\sigma$, to reproduce different signal-to-noise ratio regimes.

In the following, we compare the PC model (\ref{eq:model_f_pca}) to an alternative approach that we call the naive approach. In fact, the individual functions $\{\vect{x}_l\}$ could be estimated from $\{\vect{y}_l\}$ by use of classical inverse problem estimators. Here, we adopt the estimates $\{\hat{\vect{x}}_l\}$ defined as
\begin{equation}\label{eq:indiv_rec}
\hat{\vect{x}}_l = \argmin \limits_{\substack{\vect{f} = (f_1,f_2,f_3): \\ f_1,f_2,f_3 \in H^2(\M)}} \sum \limits_{l=1}^m \| \vect{y}_l - K \Psi \vect{f} \|^2 + \lambda \int_{\mathcal{M}} \! \|\Delta_{\mathcal{M}} \vect{f}\|^2, \qquad l=1,\ldots,m,
\end{equation}
where each $\hat{\vect{x}}_l$ is defined in such a way that it balances the fitting term and the regularization term in (\ref{eq:indiv_rec}). Due to the fact that $\vect{f}$ is vector-valued, $\| \Delta_{\mathcal{M}} \vect{f}\|^2$ is defined as
\[
\| \Delta_{\mathcal{M}} \vect{f}\|^2 = \Delta_{\mathcal{M}}^2 f_1 + \Delta_{\mathcal{M}}^2 f_2 + \Delta_{\mathcal{M}}^2 f_3,
\]
with $\{f_q: q=1,2,3\}$ denoting the components of $\vect{f}$. The same penalty operator is also adopted to generalize to vector-valued functions the PC models introduced in Sections \ref{sec:PCfunctions}-\ref{sec:PCcovariances}.  In this approach, the constant $\lambda$ is chosen independently for each of the $m$ functions by partitioning the $241$ detectors in roughly equally sized $K=2$ groups and applying $K$-fold cross-validation. The criterion for the optimal $\lambda$ is the average reconstruction error, on the sensors space, computed on the validation groups. Once we obtain the estimates $\{\hat{\vect{x}}_l\}$ we can compute the estimated PC functions $\{\bm{\psi}_r\}$ by applying classical multivariate PC analysis on the reconstructed objects $\hat{\vect{x}}_l$.

The estimates are compared to those of the proposed PC function model, as described in Algorithm~\ref{alg:InversefPCA}, with 15 iterations. Note that, instead, a tolerance could be fixed to test if the algorithm has converged. However, 15 iterations give satisfactory convergence levels in our simulations and application studies. We partition the $m$ observations in equally sized $K = 2$ groups and perform $K$-fold cross-validation for the choice of the penalty. Specifically, we choose the coefficient $\lambda$ that minimizes the sensors space reconstruction error, on the validation groups.

To evaluate the performances of the two approaches, we generate 100 datasets as previously detailed. The quality of the estimated $r$th PC function is then measured with $E(\bm{\psi}_r,\hat{\bm{\psi}}_r) = \sum_{q=1}^3 \|\nabla_{\mathcal{M}}(\psi_{r,q} - \hat{\psi}_{r,q})\|^2$. The results are summarized in the boxplots in Figure~\ref{fig:PCfunctions_sim_boxplots}, for two different signal-to-noise ratios, where the Gaussian noise has standard deviation $\sigma = 5$ and $\sigma = 10$. In Figure~\ref{fig:signal_sim} we show an example of a signal on the brain space corrupted with the specified noise levels.

The boxplots highlight the fact that the proposed approach provides better estimates of the PC functions (i.e. lower estimation errors $E(\bm{\psi}_r,\hat{\bm{\psi}}_r)$), when compared to the naive approach. Differences in the estimation error are higher in a low signal-to-noise regime, as it is for the estimation of the fourth PC function, where intuitively, the low variance associated to the PC function makes it more difficult to distinguish this structured signal from the noise component. Also surprising is the stability of the estimates of the proposed algorithm across the generated datasets, as opposed to the naive approach of reconstructing the functional observations independently, which instead returns multiple particularly unsatisfactory reconstructions. An example of such reconstructions is shown in Figure~\ref{fig:PCfunctions_sim_figure}.

\begin{figure}[!htb]
\centering
\includegraphics[width=1\textwidth]{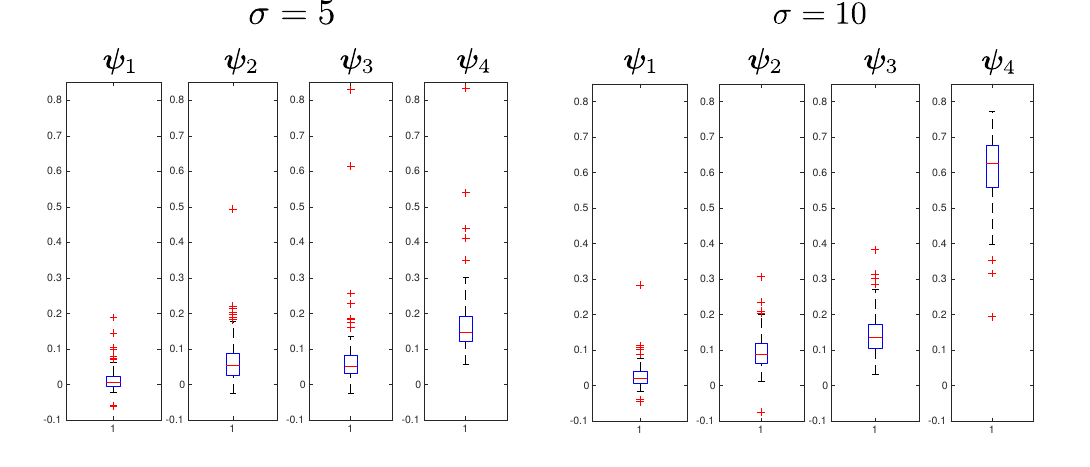}%
\caption[]{
On the left, a summary of the results in a medium signal-to-noise ratio regime. On the right, a summary of the results in a low signal-to-noise ratio regime. Each boxplot displays the paired differences of the estimation errors $E(\bm{\psi}_r,\hat{\bm{\psi}}_r)$ between the estimates of the two steps naive method and those obtained by applying Algorithm~\ref{alg:InversefPCA}. A paired difference greater than $0$ indicates that, for the dataset in question, Algorithm~\ref{alg:InversefPCA} has performed better than the two steps naive approach.}
\label{fig:PCfunctions_sim_boxplots}
\end{figure}

\begin{figure}[!htb]
\centering
\includegraphics[width=0.65\textwidth]{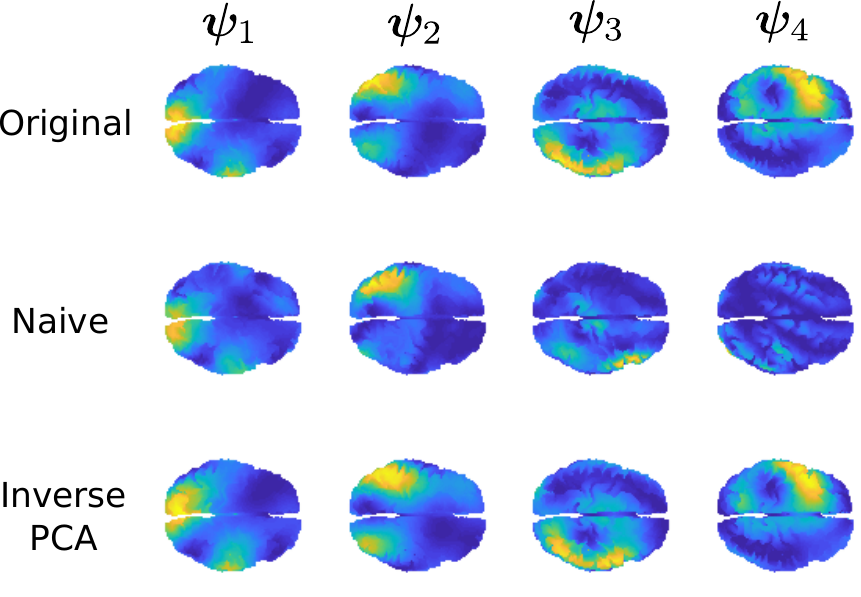}%
\caption[]{On the first row the energy maps of the true four PC components to be estimated, on the second row the estimations given by the two steps naive method, and on the third row the reconstructions obtained by applying Algorithm~\ref{alg:InversefPCA}.}
\label{fig:PCfunctions_sim_figure}
\end{figure}

\subsection{Indirectly observed covariance functions}
In this section, we consider $\M_\mathcal{T}$ to be a 8K nodes triangular mesh, this time representing a template geometry of the cortical surface, which is shown in Figure~\ref{fig:Conte69_8K}. This contains only the geometric features common to all subjects. Moreover, each subject's cortical surface is also represented by a 8K nodes triangular surface, which is used, together with the locations of the $241$ detectors on the sensors space, and the head model, to compute a forward operator $K_i$ for the $i$th subject. The 8K nodes of each subject's triangular mesh are in correspondence with the 8K nodes of the template mesh $\M_\mathcal{T}$. This allows the model to be defined on the template $\M_\mathcal{T}$.

As in the previous section, we construct four $L^2(\M_\mathcal{T})$ orthonormal functions $\{\bm{\psi}_r = (\psi_{r,1},\psi_{r,2},\psi_{r,3}): r=1,\ldots,4\}$. The energy maps of $\{\bm{\psi}_r\}$ are shown in Figure~\ref{fig:PCcovfunctions_est_sim}. We generate synthetic data from model (\ref{eq:cov_brain_sensors}) as follows:
\[
C_i = \sum_{r=1}^4 z_{ir}^2 \bm{\psi}_r\otimes \bm{\psi}_r
= \sum_{r=1}^4 z_{ir}^2
\begin{bmatrix}
	\psi_{r,1}\otimes\psi_{r,1} & \psi_{r,1}\otimes\psi_{r,2} & \psi_{r,1}\otimes\psi_{r,3}\\
	\psi_{r,2}\otimes\psi_{r,1} & \psi_{r,2}\otimes\psi_{r,2} & \psi_{r,2}\otimes\psi_{r,3}\\
	\psi_{r,3}\otimes\psi_{r,1} & \psi_{r,3}\otimes\psi_{r,2} & \psi_{r,3}\otimes\psi_{r,3}
\end{bmatrix},
\]
where $z_{i1}$, \ldots, $z_{i4}$ are i.i.d realizations of the four independent random variables $\{z_r \sim N(0, \gamma_r): r=1,\ldots,4\}$, with $\gamma_1 = 3^2$, $\gamma_2 = 2.5^2$,  $\gamma_3 = 2^2$ and $\gamma_4 = 1$. The matrix-valued form of the covariance functions arises from the fact that the observed functions on the brain space are vector-valued. Subsequently, we construct the point-wise evaluations matrices $\mathbb{C}_i \in \R^{24K \times 24K}$, from which the correspondent covariance matrices on the sensors space are defined as
\[
S_i = K_i \mathbb{C}_i K_i^T + E_i^T E_i, \qquad i=1,\ldots,n.
\]
\begin{figure}[!htb]
\centering
\includegraphics[width=0.75\textwidth]{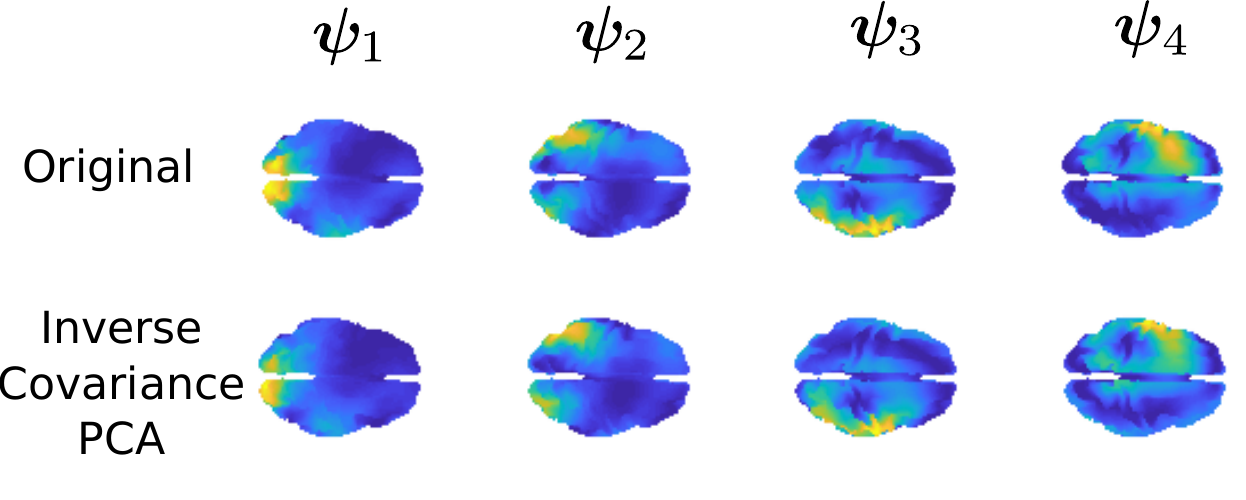}%
\caption[]{On the top row, the energy maps of $\bm{\psi}_1, \ldots, \bm{\psi}_4$. On the bottom row the energy maps of the estimates $\hat{\bm \psi}_1, \ldots, \hat{\bm \psi}_4$ obtained by applying Algorithm~\ref{alg:InverseCovfPCA}.}
\label{fig:PCcovfunctions_est_sim}
\end{figure}
The term $E_i^T E_i$ is an error term, where $E_i$ is a $s \times s$ matrix with each entry that is an independent sample from a Gaussian distribution with mean zero and standard deviation $5$. We then apply Algorithm~\ref{alg:InverseCovfPCA} with 15 iterations, feeding in input $\{S_i\}$. The results are shown in Figure~\ref{fig:PCcovfunctions_est_sim}, in terms of energy maps of the reconstructed functions $\{\hat{\bm{\psi}}_r\}$. These are a close approximation of the underlying functions $\{\bm{\psi}_r\}$. The fidelity measure $\sum_{q=1}^3 \|\nabla_{\mathcal{M}}(\psi_{r,q} - \hat{\psi}_{r,q})\|^2$ of such estimates is $6.8 \times 10^{-2}$, $6.1 \times 10^{-1}$, $6.8 \times 10^{-1}$ and $7.4 \times 10^{-1}$, for $\bm{\psi}_1,\ldots,\bm{\psi}_4$ respectively, which is comparable in term of order of magnitude to the results obtained in the case of PCs of indirectly observed functions. Across the generation of multiple datasets, results are stable, with the exception of few situations where the cross-validation approach suggests a penalization coefficient $\lambda$ that under-smoothes the solution, due to very similar associated signals on the sensors space of the under-smoothed solution and the real solution. However, the cross-validation is only a possible approach to the choice of the penalization constant, and many other options have been proposed in the inverse problems literature, \citep[see, e.g.,][]{Vogel2002}. Some of these, however, involve visual inspection.

\section{Application}\label{sec:application}
In this section, we apply the developed models to the publicly available Human Connectome Project (HCP) Young Adult dataset \citep{Essen2012}. This dataset comprises multi-modal neuroimaging data such as structural scans, resting-state and task-based functional MRI scans, and resting-state and task-based MEG scans from a large number of healthy volunteers. In the following, we briefly review the pre-processing pipeline, applied to such data by the HCP, to ultimately facilitate their use.
\subsection{Pre-processing}
For each individual a high-resolution 3D structural MRI scan has been acquired. This returns a 3D image  describing the structure of the gray and white matter in the brain. Gray matter is the source of large parts of our neuronal activity. White matter is made of axons connecting the different parts of the gray matter. If we exclude the sub-cortical structures, gray matter is mostly distributed at the outer surface of the cerebral hemispheres. This is also known as the cerebral cortex.

\begin{figure}[!htb]
\centering
\includegraphics[width=0.55\textwidth]{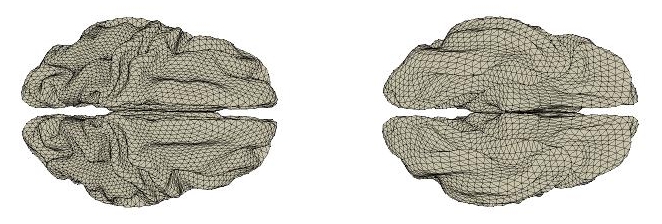}%
\caption[]{Top side and bottom side views of the template triangular mesh $\M_\mathcal{T}$ composed of 8K nodes.}
\label{fig:Conte69_8K}
\end{figure}

By segmentation of the 3D structural MRI, it is possible to separate gray matter from white matter, in order to extract the cerebral cortex structure. Subsequently a mid-thickness surface, interpolating the mid-points of the cerebral cortex, can be estimated, resulting in a 2D surface embedded in a 3D space that represents the geometry of the cerebral cortex. In practice, such a surface, sometimes referred to as cortical surface, is a triangulated surface. Moreover, from the 3D structural MRI, a surface describing the individuals' head can be extracted. The latter plays a role in the derivation of the model for the electrical/magnetic propagation of the signal from the cerebral cortex to the sensors. An example of the cortical surface of a single subject, is shown on the right panel in Figure~\ref{fig:forwardoperator}, instead the associated head surface and MEG sensors positions are shown on the left panel of the same figure.

Moreover, a surface based registration algorithm has been applied to register each of the extracted cortical surfaces to a triangulated template cortical surface, which is shown in Figure~\ref{fig:Conte69_8K}. Post registration, the triangulated template cortical surface is sub-sampled to a 8K nodes surface. Moreover, the nodes on the cortical surface of each subject are also sub-sampled to a set of 8K nodes in correspondence to the 8K nodes of the template. For each subject, a $248 \times 24K$ matrix, representing the forward operator, has been computed with FieldTrip \citep{Oostenveld2011} from its head surface, cortical surface and sensors position. Such a matrix relates the vector-valued signals in $\R^3$, on the nodes of the triangulation of the cerebral cortex, to the one detected from the sensors, consisting of 248 magnetometer channels.

With the aim of studying the functional connectivity of the brain, for each subject, three 6 minutes resting state MEG scans have been performed, of which one session is used in our analysis. During the 6 minutes, data are collected from the sensors at 600K uniformly distributed time-points.  Using FieldTrip, classical pre-processing is applied to the detected signals, such as low quality channels and low quality segments removal. Details of this procedure can be found in the HCP MEG Reference Manual. Moreover, we apply a band pass filter, limiting the spectrum of the signal to the $[12.5, 29]$Hz, also known as the beta waves. For the signal of each channel we compute its amplitude envelope (see Figure~\ref{fig:hilb_env}) which describes the evolution of the signal amplitude. The measure of connectivity between channels that we adopt in this work is the covariance of the amplitude envelopes. Other connectivity metrics, such as phase-based metrics, have been proposed in the literature \citep[see, e.g. ][and references therein]{Colclough2016}.

\begin{figure}[!htb]
\centering
\includegraphics[width=0.8\textwidth]{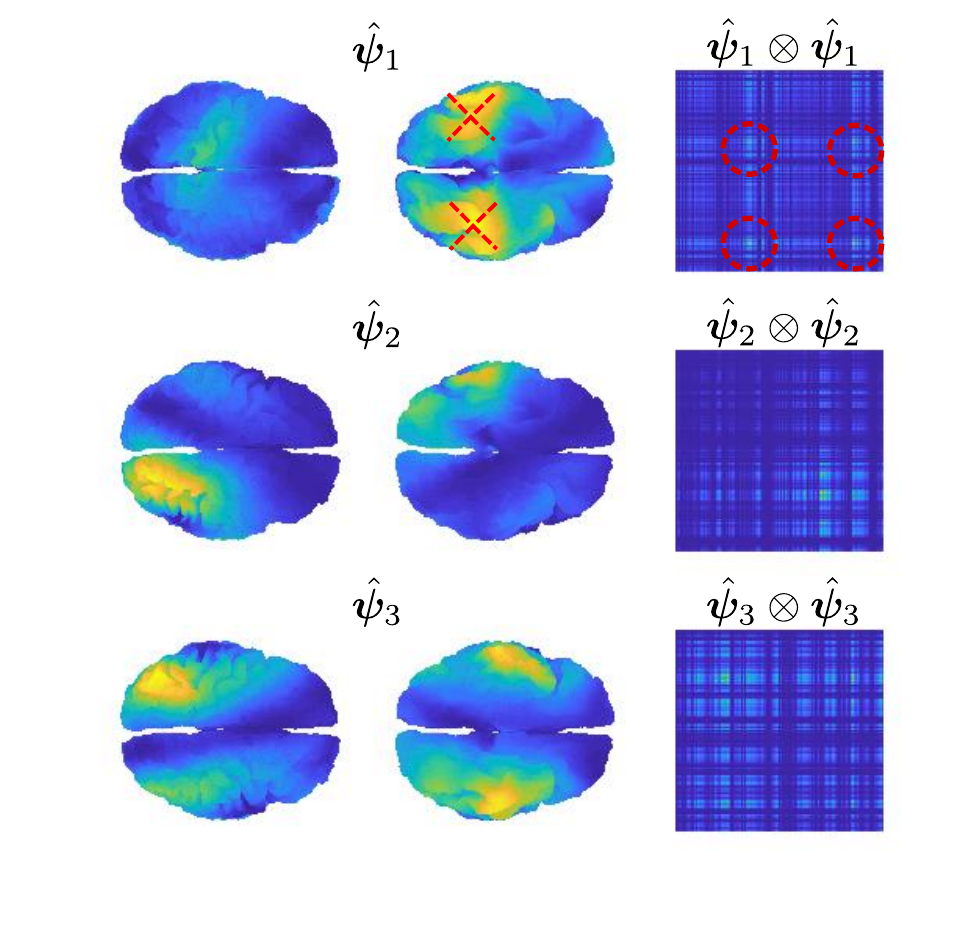}%
\caption[]{Top side and bottom side views of the estimated energy maps $\hat{\bm{\psi}}_1$, $\hat{\bm{\psi}}_2$ and $\hat{\bm{\psi}}_3$ obtained by applying Algorithm~\ref{alg:InverseCovfPCA} to the covariance matrices computed from the MEG resting state data of a single subject on $n=40$ consecutive time intervals. On the right panel, the covariance functions associated with these energy maps. On the top right panel we highlight with red circles the areas with high average interconnectivity, which correspond to the neighborhoods of the red crossed vertices in the plot of the energy map of $\bm \psi_1$.}
\label{fig:PCcovfunctions_dyn}
\end{figure}

\begin{figure}[!htb]
\centering
\includegraphics[width=0.8\textwidth]{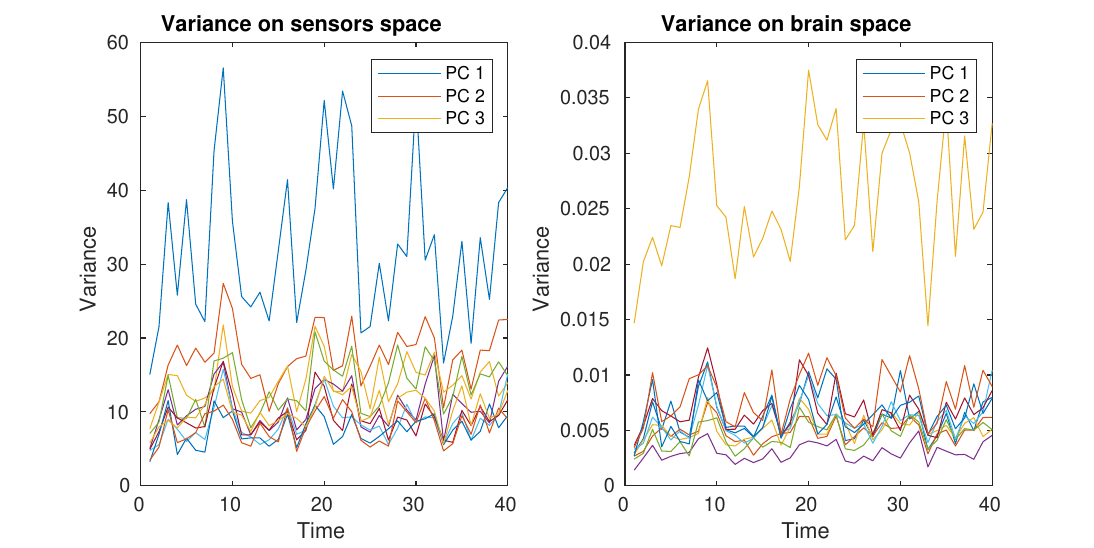}%
\caption[]{Plots of the segment-specific variances of the first $R=10$ PC covariance functions. On the left, the estimated variances on the sensors space, on the right, the estimated variances on the brain space.}
\label{fig:scores_covfunctions_dyn}
\end{figure}

\subsection{Analysis}
Here we apply the population model introduced in Section~\ref{sec:PCcov_pop} to the HCP MEG data. The first part of the analysis focuses on studying dynamic functional connectivity of a specific subject. For this purpose, we subdivide the 6 minutes session in $n=40$ consecutive intervals. Each of these segments is used to compute a covariance matrix in the sensors space, resulting in $n$ covariance matrices $S_1, \ldots, S_{n}$. In this setting, we have one forward operator $K = K_1 = \ldots = K_n$. The aim is understanding the main modes of variation of the functional connectivity on the brain space of the subject. Thus, Algorithm~\ref{alg:InverseCovfPCA}, with $20$ iterations, is applied to $S_1, \ldots, S_{n}$ to find the PC covariance functions. 

A regularization parameter $\lambda$ common to all the PC components is chosen by inspecting the plot of the regularity of the first $R=10$ PC covariance functions ($\sum_{r=1}^{R} \int_\M \|\nabla \bm \psi_r\|^2$) versus the residual norm, for different choices of the parameter. This is a version of the L-curve plot \citep{Hansen2000} and is shown on the left panel of Figure~\ref{fig:lcurves}. Here we show the results for $\lambda = 10^2$, in the appendices we show the results for $\lambda = 10$. The energy maps of the estimated $\hat{\bm{\psi}}_1$, $\hat{\bm{\psi}}_2$ and $\hat{\bm{\psi}}_3$ resulting from the analysis are shown in Figure~\ref{fig:PCcovfunctions_dyn}. These are associated with the first three PC covariance functions $\hat{\bm{\psi}}_1 \otimes \hat{\bm{\psi}}_1$, $\hat{\bm{\psi}}_2 \otimes \hat{\bm{\psi}}_2$ and $\hat{\bm{\psi}}_3 \otimes \hat{\bm{\psi}}_3$. High intensity areas, in yellow, indicate which areas present high average interconnectivity, either by means of positive or negative correlation in time. 

In Figure~\ref{fig:scores_covfunctions_dyn}, we show the plot of variances associated with each time segment, describing the variation in time of the PC covariance functions, hence the variation in interconnectivity. The variance can be either defined on the sensors space, by normalizing the PC covariance functions $\{K \hat{\bm{\psi}}_r\}$, with $K$ the forward operator, or on the brain space, by normalizing the PC covariance functions on the brain space $\{\hat{\bm{\psi}}_r\}$. Due to the presence of invisible dipoles, which are dipoles that display zero magnetic field on the sensors space, the two norms can be quite different, leading to different average variances for each PC covariance function. Due to the high sensitivity of the source space variances on the choice of the regularization parameter, we focus on the estimated variances on the sensors space.

We have also applied our model to the covariances obtained by subdividing the MEG session in $n=80$ segments. As expected the PC covariance functions, shown in Figure~\ref{fig:PCcovfunctions_dyn_n80} are very similar. However, the variances, in Figure~\ref{fig:scores_covfunctions_dyn_loglambda1}, show higher variability in time, which can be partially explained by the fact that shorter time segments lead to covariance estimates that have higher variability. 

\begin{figure}[!htb]
\centering
\includegraphics[width=0.8\textwidth]{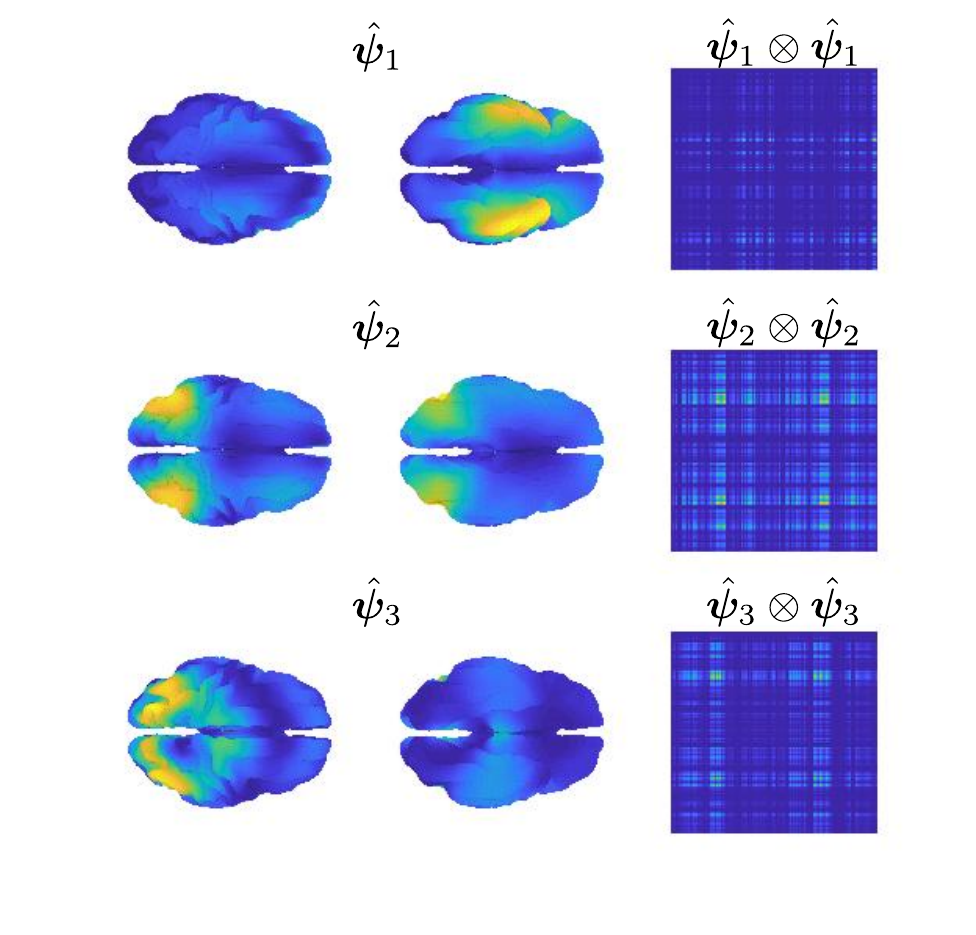}%
\caption[]{Top side and bottom side views of the estimated energy maps $\hat{\bm{\psi}}_1$, $\hat{\bm{\psi}}_2$ and $\hat{\bm{\psi}}_3$ obtained by applying Algorithm~\ref{alg:InverseCovfPCA} to the covariance matrices computed from the MEG resting state data of $n=40$ different subjects. On the right panel, the covariance functions associated with these energy maps.}
\label{fig:PCcovfunctions_multi}
\end{figure}

\begin{figure}[!htb]
\centering
\includegraphics[width=0.8\textwidth]{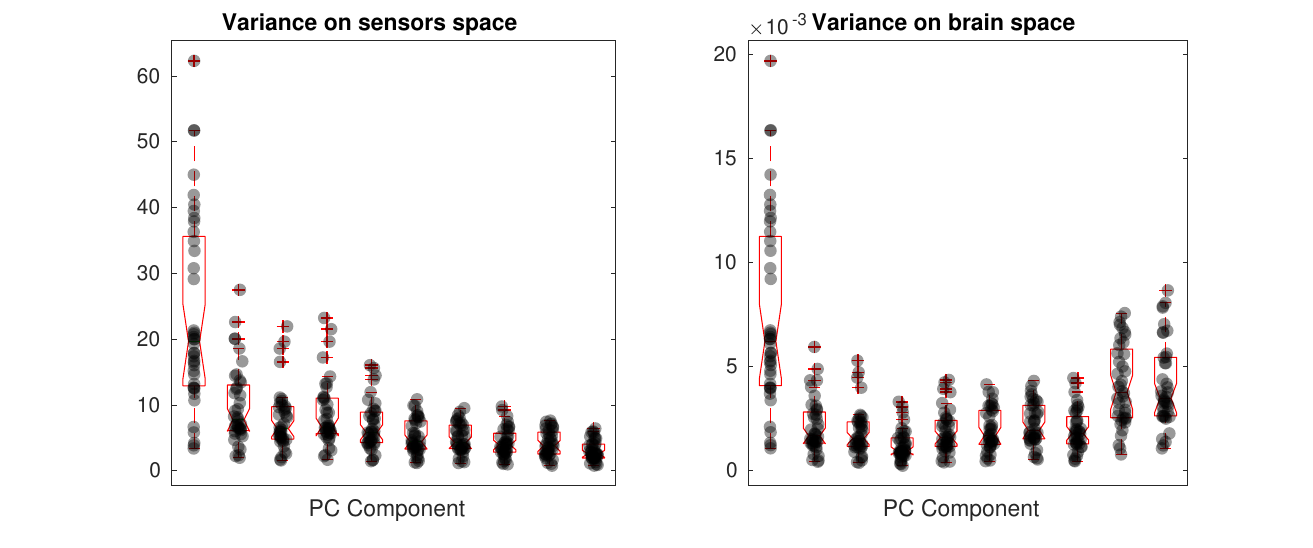}%
\caption[]{Plots of the subject-specific variances associated with the first $R = 10$ PC covariance functions. On the left, the estimated variances on the sensors space, on the right, the estimated variances on the brain space.}
\label{fig:scores_covfunctions_multi}
\end{figure}

The second part of the analysis focuses on applying the proposed methodology to a multi-subject setting. Specifically, $n=40$ different subjects are considered. For each subject, the 6 minutes scan is used to compute a covariance matrix, resulting in $n$ covariance matrices $S_1, \ldots, S_{n}$. The template geometry in Figure~\ref{fig:Conte69_8K} is used as a model of the brain space. Algorithm~\ref{alg:InverseCovfPCA} is then applied to find the PC covariance functions on the template brain, associated with $S_1, \ldots, S_{n}$. We run the algorithm for $20$ iterations, and choose the regularizing parameter to be $\lambda=10^2$ by inspecting the L-curve plot in the right panel of Figure~\ref{fig:lcurves}. The results for $\lambda=10$ are shown in the appendices. The energy maps of the estimated functions $\hat{\bm{\psi}}_1$, $\hat{\bm{\psi}}_2$ and $\hat{\bm{\psi}}_3$ and the associated first three covariance functions $\hat{\bm{\psi}}_1 \otimes \hat{\bm{\psi}}_1$, $\hat{\bm{\psi}}_2 \otimes \hat{\bm{\psi}}_2$ and $\hat{\bm{\psi}}_3 \otimes \hat{\bm{\psi}}_3$, are shown in Figure~\ref{fig:PCcovfunctions_multi}. High intensity areas, in yellow, indicate which areas present high average connectivity. In Figure~\ref{fig:scores_covfunctions_multi}, we show the subject-specific associated variances, both in the sensors space and the brain space.

The presented methodology opens up the possibility to understand population level variation in functional connectivity, and indeed, whether, just as we need different forward operators for individuals (due to anatomical differences), we should also be considering both population and subject-specific connectivity maps when analyzing connectivity networks. In fact, it is of interest to note that in both the single and multi-subject settings, the areas with high interconnectivity, displayed in yellow in Figure~\ref{fig:PCcovfunctions_dyn} and \ref{fig:PCcovfunctions_multi}, seem to be at least partially overlapping with the brain's default network \citep{Buckner2008, Yeo2011}. The brain's default network consists of the brain regions known to have highly correlated hemodynamic activity (i.e.  highest functional connectivity levels), and to be most active, when the subject is not performing any specific task. An image of the spatial configuration of the default network can be found, for instance, in Figure 2 of \cite{Buckner2008}. From the plots of the associated variances in the sensors space (left panel of Figure~\ref{fig:scores_covfunctions_dyn} and Figure~\ref{fig:scores_covfunctions_multi}) we can see that these areas are also the ones that show high variability in connectivity across time or across subjects. This might suggest that the brain's default network is also the brain region that shows among the highest levels of spontaneous variability in connectivity.

The plots of the variances on the brain space (right panel of Figure~\ref{fig:scores_covfunctions_multi}), when compared to those on the sensors space (left panel of Figure~\ref{fig:scores_covfunctions_multi}), demonstrate that these type of studies are highly sensitive to the choice of the regularization, not only in terms of spatial configuration of the results, but also in terms of estimated variances on the brain space. With a naive `first reconstruct and then analyze' approach, where the reconstructed data on the brain space replace those observed on the sensors space, this issue could go unnoticed, as the variability that does not fit the chosen model is implicitly discarded in the reconstruction step and does not appear in the subsequent analysis. Also, importantly, our analysis deals with statistical samples that are entire covariances, overcoming the limitations of seed-based approaches, where prior spatial information is required to choose the seed. Seed locations are usually informed by fMRI studies and this comes with the risk of biasing the analysis when comparing electrophysiological networks (MEG) and hemodynamic networks (fMRI).

In general, care should be taken when drawing conclusions from MEG studies. Establishing static and dynamic functional connectivity from MEG data remains challenging, due to the strong ill-posedness of the inverse problem. It is known that other variables, such as the choice of the frequency band or the choice of the connectivity metric can influence the analysis. While the choice of the neural oscillatory frequency band could be seen as an additional parameter in MEG functional connectivity studies, there is no general agreement on the choice of the connectivity metrics \citep{Gross2013}. It is important to highlight that in this paper we focus on methodological contributions to the specific problem of reconstructing and representing indirectly observed functional images and covariance functions. 

\section{Discussion}\label{sec:discussion}
In this work we introduce a general framework for the reconstruction and representation of covariance operators in an inverse problem context. We first introduce a model for indirectly observed functional images in an unconstrained space, which outperforms the naive approach of solving the inverse problem individually for each sample. This model plays an important role in the case of samples that are indirectly observed covariance functions, and thus constrained to be positive semidefinite. We deal with the non-linearity introduced by such constraint by working with unconstrained representations, yet incorporating spatial information in their estimation. The proposed methodology is finally applied to the study of brain connectivity from the signals arising from MEG scans.

The models proposed here can be extended in many interesting directions. From an applied prospective, it is of interest to apply them to different settings, not necessarily involving neuroimaging, where studying second order information has been so far prohibitive. Direct examples are second order analysis of the dynamics of meteorological observations, such as temperature. Another possible application is the study of the dynamics of ocean currents, where the irregularity of the spatial domain, and its complex boundaries, can be easily accounted for thanks to the manifold representation approach in our models.

From a modeling point of view, it is of interest to take a step further towards the integration of the inverse problems literature with the approach we adopt in this paper. For instance, penalization terms that have been shown to be successful in the inverse problems literature, e.g. total variation penalization, could be introduced in our models.


\appendix
\appendixpage
\renewcommand\thefigure{\thesection.\arabic{figure}}

\section{Discrete solutions}
\begin{proof}[Proof of Proposition \ref{prop:FE_func}]
We want to find a minimizer $\hat{f} \in H^2(\M)$, given $\vect{z}$ with $\|\vect{z}\|=1$, of the objective function in (\ref{eq:model_f_pca}):
\begin{align}
&\sum \limits_{l=1}^m \| \vect{y}_l - z_l K_l \Psi f \|^2 + \lambda \vect{z}^T\vect{z} \int_{\mathcal{M}} \! \Delta^2_{\mathcal{M}} f \nonumber \\
&\propto (\Psi f)^T (\sum \limits_{l=1}^m z_l^2  K_l^T K_l) \Psi f -2 (\Psi f)^T (\sum \limits_{l=1}^m z_l K_l^T \vect{y}_l )  + \lambda \int_{\mathcal{M}} \! \Delta^2_{\mathcal{M}} f. \label{eq:objective_function_func}
\end{align}
An equivalent formulation of a minimizer $\hat{f} \in H^2(\M)$ of such objective function is given by satisfying the equation
\begin{equation}\label{eq:lagrange}
(\Psi \varphi)^T (\sum \limits_{l=1}^m z_l^2  K_l^T K_l) \Psi \hat{f} + \lambda \int_{\mathcal{M}} \! \Delta_{\mathcal{M}} \varphi \Delta_{\mathcal{M}} \hat{f} = (\Psi \varphi)^T (\sum \limits_{l=1}^m z_l K_l^T \vect{y}_l )
\end{equation}
for every $\varphi \in H^2(\M)$ \citep[see][Chapter 2]{Braess2007}. Moreover, such minimizer is unique if $A(\varphi, f) = (\Psi \varphi)^T (\sum \limits_{l=1}^m z_l^2  K_l^T K_l) \Psi f + \lambda \int_{\mathcal{M}} \! \Delta_{\mathcal{M}} \varphi \Delta_{\mathcal{M}} f$ is positive definite. Given that for a closed manifold $\M$, $\int_{\mathcal{M}} \! \Delta_{\mathcal{M}}^2 f = 0$ iff $f$ is a constant function \citep{Dziuk2013}, the positive definiteness condition is equivalent to assuming that $\ker(\sum \limits_{l=1}^m z_l^2  K_l^T K_l)$, the kernel of $\sum \limits_{l=1}^m z_l^2  K_l^T K_l$, does not contain the subspace of $p$-dimensional constant vectors.

Moreover, we can reformulate equation (\ref{eq:lagrange}) in a form that involves only first-order derivatives by integration by parts against a test function. We then look for a solution in the discrete space $V \subset H^1(\M)$, i.e. finding $\hat{f},g \in V$
\begin{align}\label{eq:lagrange_aux_discrete}
\begin{cases}
&(\Psi \varphi)^T (\sum \limits_{l=1}^m z_l^2  K_l^T K_l) \Psi \hat{f} + \lambda \int_{\mathcal{M}} \! \nabla_{\mathcal{M}} \varphi \cdot \nabla_{\mathcal{M}} g = (\Psi \varphi)^T (\sum \limits_{l=1}^m z_l K_l^T \vect{y}_l ) \\
&\int_{\mathcal{M}} \! \nabla_{\mathcal{M}} \hat{f} \cdot \nabla_{\mathcal{M}} w  - \int_{\mathcal{M}} \! g w = 0
\end{cases}
\end{align}
for all $\varphi,w \in V$. The operator $\nabla_{\mathcal{M}}$ is the gradient operator on the manifold $\M$. The gradient operator $\nabla_{\mathcal{M}}$ is such that $(\nabla_{\mathcal{M}} w)(v)$, for $w$ a smooth real function on $\M$ and $v \in \M$, takes value on the tangent space at $v$. We denote with $\cdot$ the scalar product on the tangent space.

We recall here the definition of the $\kappa \times \kappa$ matrices to be $(M)_{jj'} =\int_{\mathcal{M}_{\mathcal{T}}} \phi_j \phi_{j'}$ and $(A)_{jj'}=\int_{\mathcal{M}_\mathcal{T}} \nabla_{\mathcal{M}_\mathcal{T}} \phi_j \cdot \nabla_{\mathcal{M}_\mathcal{T}} \phi_{j'}$. Note that requiring (\ref{eq:lagrange_aux_discrete}) to hold for all $\varphi,w \in V$ is equivalent to requiring that (\ref{eq:lagrange_aux_discrete}) holds for all $\varphi,w$ that are basis elements of $V$, thus exploiting the basis expansion formula (\ref{eq:basis}) we can characterize (\ref{eq:lagrange_aux_discrete}) with the solution of the linear system
\begin{equation}\label{eq:linear_system}
	\begin{bmatrix}
		\sum \limits_{l=1}^m z_l^2  K_l^T K_l& \lambda A\\
		A& - M
	\end{bmatrix}
	\begin{bmatrix}
		\vect{\hat{c}}\\
		\vect{\hat{q}}
	\end{bmatrix}
=
	\begin{bmatrix}
		\sum \limits_{l=1}^m z_l K_l^T \vect{y}_l\\
		\vect{0}
	\end{bmatrix},
\end{equation}
where $\vect{\hat{c}}$ and $\vect{\hat{q}}$ are the basis coefficients of $f \in V$ and $g \in V$, respectively.
Solving (\ref{eq:linear_system}) in $\vect{\hat{c}}$ leads to
\begin{equation}
(\sum \limits_{l=1}^m z_l^2  K_l^T K_l + \lambda A M^{-1} A) \vect{\hat{c}} = \sum \limits_{l=1}^m z_l K_l^T \vect{y}_l.
\end{equation}

\end{proof}

\begin{proof}[Proof of Proposition \ref{prop:FE_cov}]
We want to find a minimizer $\hat{f} \in H^2(\M)$, given $\{\vect{z}_i\}$ with $\sum_{i=1}^n \| \vect{z}_i\|^2 = 1$, of the objective function in (\ref{eq:model_cov_pca}):
\begin{align}
&\sum \limits_{i=1}^n \|S_i^{\nicefrac{1}{2}} - \vect{z}_i (K_i \Psi f)^T \|^2 + \lambda \sum \limits_{i=1}^n \| \vect{z}_i\|^2 \int_{\mathcal{M}} \! \Delta^2_{\mathcal{M}} f \nonumber \\
& \propto   (\Psi f)^T (\sum \limits_{i=1}^n \|\vect{z}_i\|^2  K_i^T K_i) \Psi f -2  (\Psi f)^T \sum \limits_{i=1}^n K_i^T S_i^{\nicefrac{T}{2}} \vect{z}_i. \label{eq:objective_function_cov}
\end{align}

Comparing (\ref{eq:objective_function_cov}) with (\ref{eq:objective_function_func}) it is evident that by following the same steps of the proof of Proposition \ref{prop:FE_func} we obtain the desired result, which is
\[
\vect{\hat{c}} = \bigg(\sum_{i=1}^n \|\vect{z}_i \|^2 K_i^TK_i + \lambda A M^{-1} A \bigg)^{-1} \sum_{i=1}^n K_i^T S_i^{\nicefrac{T}{2}} \vect{z}_i.
\]

\end{proof}

\section{Application - additional material}
\setcounter{figure}{0}

Here we present further material complementing the analysis in Section~\ref{sec:application}. In Figure~\ref{fig:hilb_env} we show the amplitude envelope computed from a filtered version of a signal detected by an MEG sensor. The covariance of the amplitude envelopes across different sensors is the measure of connectivity used in this work.

In Figure~\ref{fig:lcurves} we show the L-curve plots associated with the PC covariance models applied to the dynamic and multi-subject functional connectivity studies.

\begin{figure}[!htb]
\centering
\includegraphics[width=0.55\textwidth]{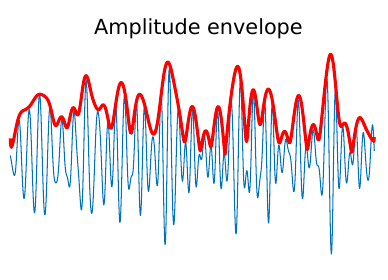}%
\caption[]{Amplitude envelope (in red) of the filtered signal (in blue) detected by an MEG sensor.}
\label{fig:hilb_env}
\end{figure}

\begin{figure}[!htb]
\centering
\includegraphics[width=0.4\textwidth]{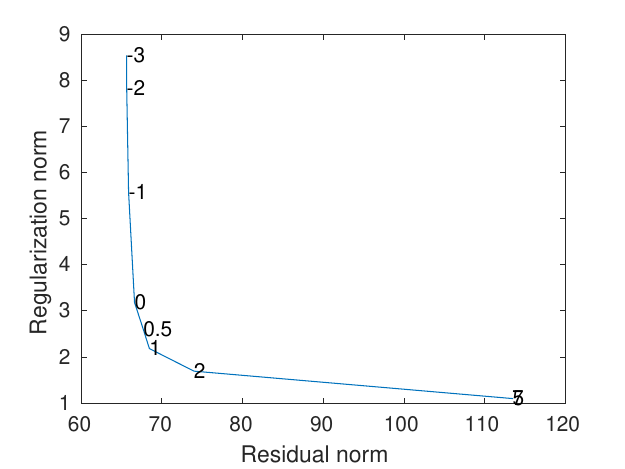}%
\includegraphics[width=0.4\textwidth]{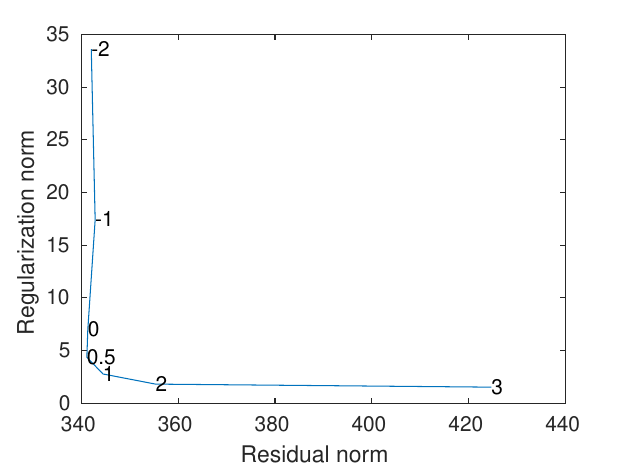}%
\caption[]{Plots of the regularity of the first $R=10$ PC covariance functions, measured as $\sum_{r=1}^{10} \int_\M \|\nabla \bm \psi_r\|^2$ versus the residual norm in the data, for different choices of $\text{log}(\lambda)$. On the left panel, the plot refers to the dynamic connectivity study, on the right panel the plot of the multi-subject connectivity study.}
\label{fig:lcurves}
\end{figure}

In Figure~\ref{fig:PCcovfunctions_dyn_loglambda1}-\ref{fig:scores_covfunctions_dyn_loglambda1} we show respectively the plots of the estimated PC covariance functions and associated variances from the dynamic functional connectivity study on $n=40$ segments with regularization parameter $\lambda=10$.

\begin{figure}[!htb]
\centering
\includegraphics[width=0.6\textwidth]{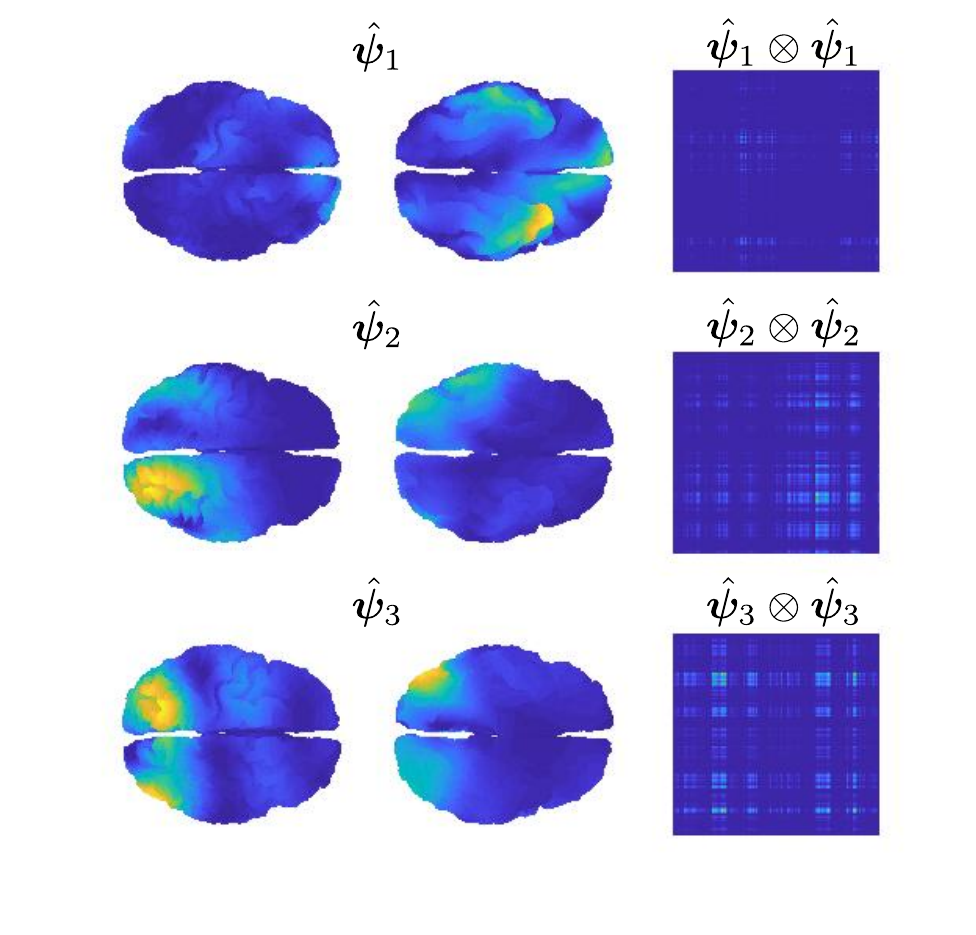}%
\caption[]{Energy maps of the estimated $\hat{\bm{\psi}}_1$, $\hat{\bm{\psi}}_2$ and $\hat{\bm{\psi}}_3$ obtained by applying Algorithm~\ref{alg:InverseCovfPCA}, with lower regularization ($\lambda=10$), to the covariance matrices computed from the MEG resting state data of a single subject on $n=40$ consecutive time intervals. On the right panel, the covariance functions associated with these energy maps.}
\label{fig:PCcovfunctions_dyn_loglambda1}
\end{figure}

\begin{figure}[!htb]
\centering
\includegraphics[width=0.8\textwidth]{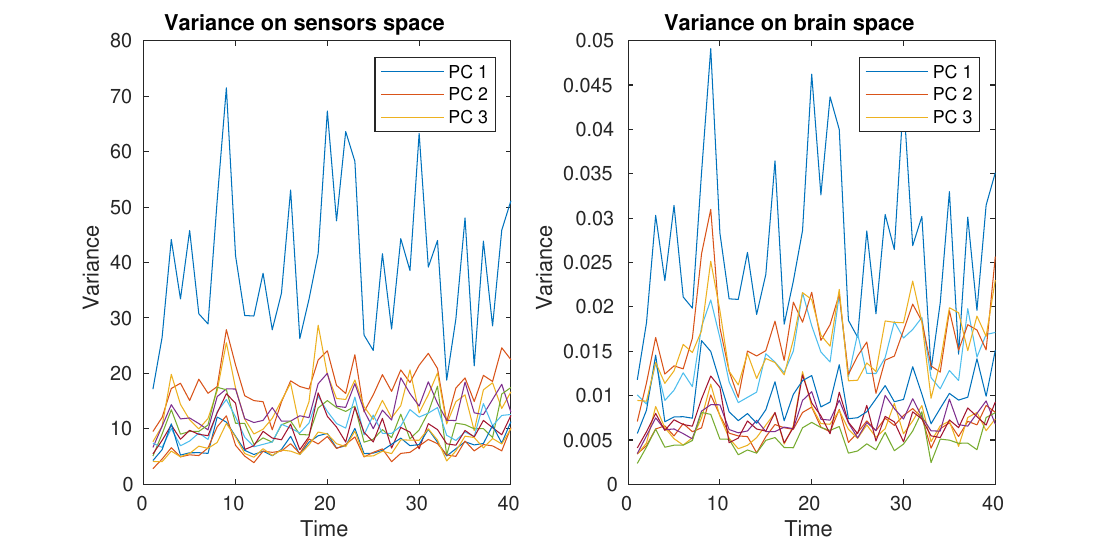}%
\caption[]{Plots of the segment-specific variances of the first $R=10$ PC covariance functions in time when a smaller regularization parameter is chosen ($\lambda=10$).}
\label{fig:scores_covfunctions_dyn_loglambda1}
\end{figure}

In Figure~\ref{fig:PCcovfunctions_dyn_n80}-\ref{fig:scores_covfunctions_dyn_n80} we show the estimated PC covariance functions and associated variances from the dynamic functional connectivity study on $n=80$ time segments with regularization parameter $\lambda=10^2$.

\begin{figure}[!htb]
\centering
\includegraphics[width=0.6\textwidth]{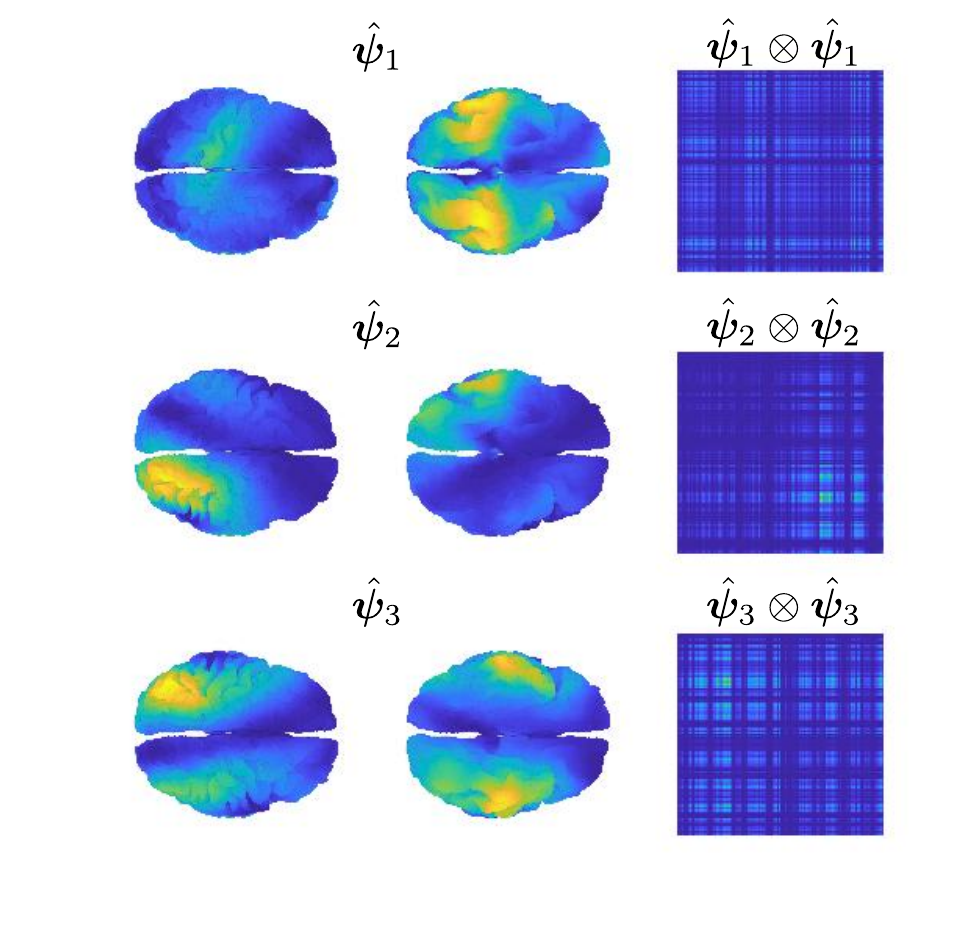}%
\caption[]{Energy maps of the estimated $\hat{\bm{\psi}}_1$, $\hat{\bm{\psi}}_2$ and $\hat{\bm{\psi}}_3$ obtained by applying Algorithm~\ref{alg:InverseCovfPCA}, with $\lambda=10^2$, to the covariance matrices computed from the MEG resting state data of a single subject on $n=80$ consecutive time intervals. On the right panel, the covariance functions associated with these energy maps.}
\label{fig:PCcovfunctions_dyn_n80}
\end{figure}

\begin{figure}[!htb]
\centering
\includegraphics[width=0.8\textwidth]{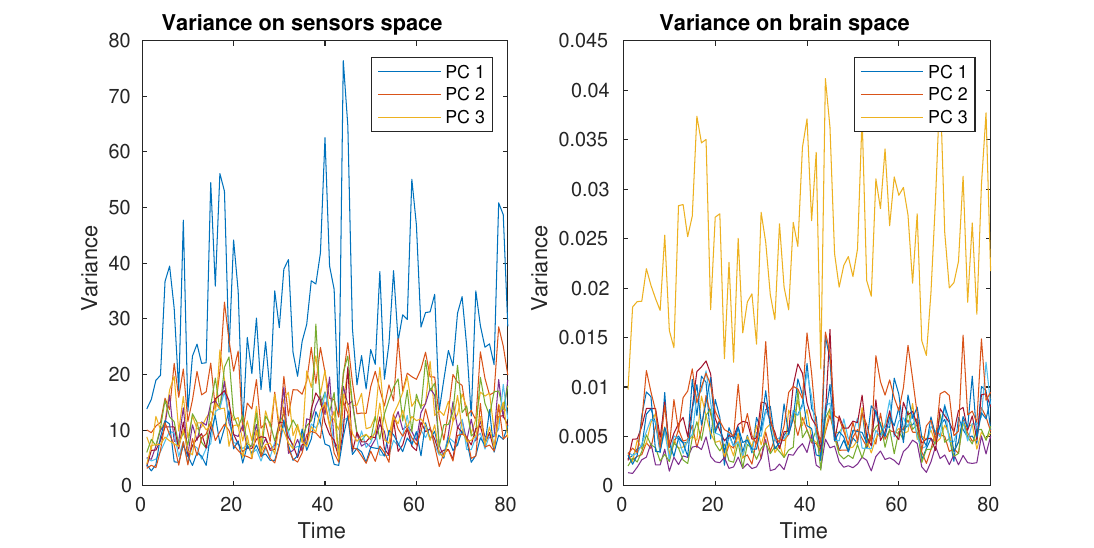}%
\caption[]{Plots of the segment-specific variances of the first $R=10$ PC covariance functions in time, with $\lambda=10^2$, when the MEG resting state data is split into $n=80$ consecutive time intervals.}
\label{fig:scores_covfunctions_dyn_n80}
\end{figure}

In Figure~\ref{fig:PCcovfunctions_multi_loglambda1}-\ref{fig:scores_covfunctions_multi_loglambda1} we show the estimated PC covariance functions and associated variances from the multi-subject functional connectivity study on $n=40$ subjects with regularization parameter $\lambda=10$.

\begin{figure}[!htb]
\centering
\includegraphics[width=0.6\textwidth]{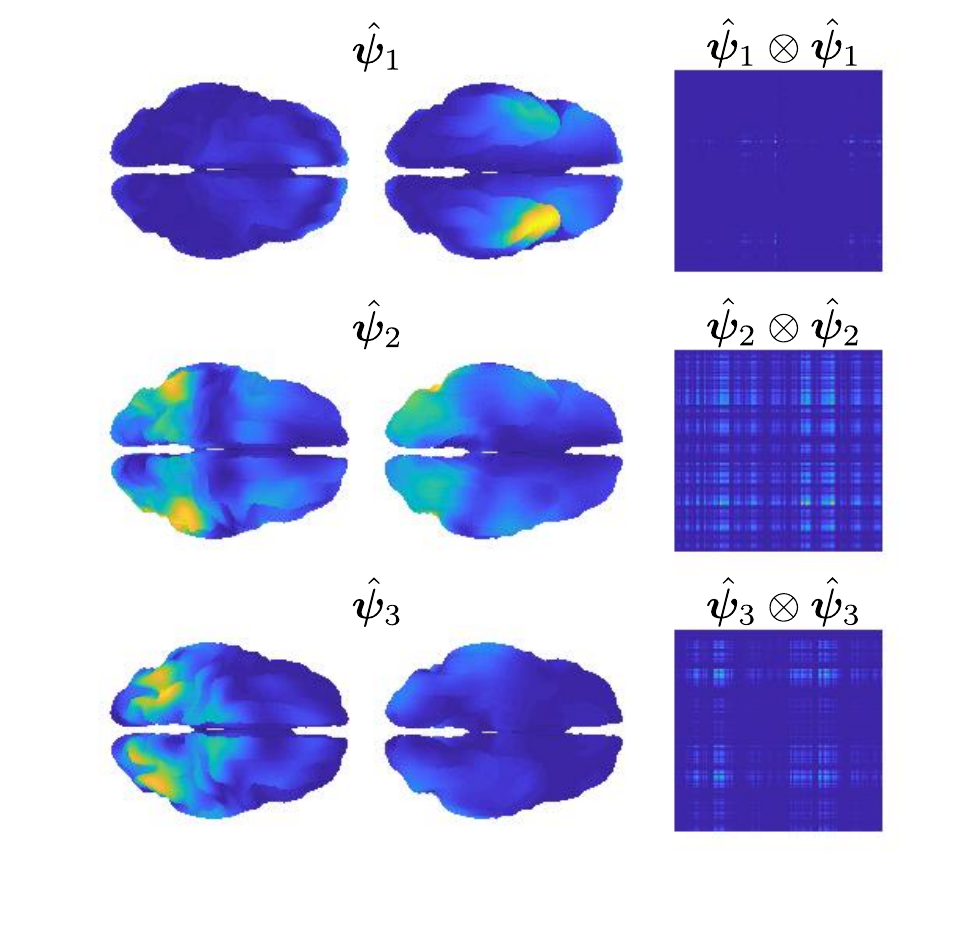}%
\caption[]{Energy maps of the estimated $\hat{\bm{\psi}}_1$, $\hat{\bm{\psi}}_2$ and $\hat{\bm{\psi}}_3$ obtained by applying Algorithm~\ref{alg:InverseCovfPCA}, with lower regularization ($\lambda=10$), to the covariance matrices computed from the MEG resting state data of $n=40$ different subjects. On the right panel, the covariance functions associated with these energy maps.}
\label{fig:PCcovfunctions_multi_loglambda1}
\end{figure}

\begin{figure}[!htb]
\centering
\includegraphics[width=0.8\textwidth]{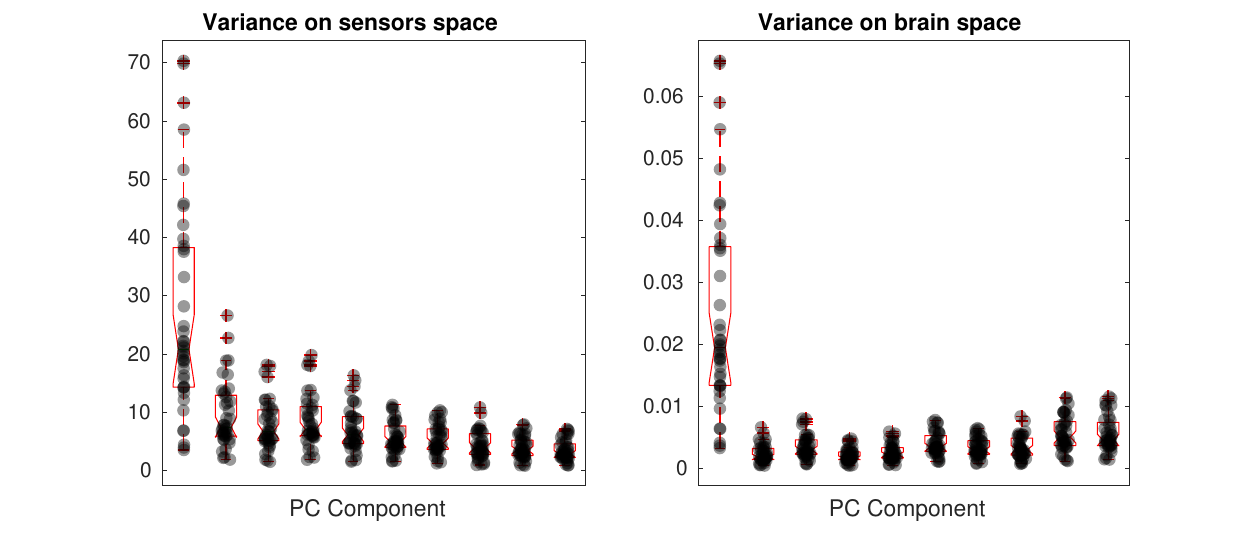}%
\caption[]{Plots of the subject-specific variances associated with the first $R=10$ PC covariance functions computed from $n=40$ subject, with regularization parameter $\lambda=10$.}
\label{fig:scores_covfunctions_multi_loglambda1}
\end{figure}

\clearpage

\textbf{Acknowledgments.} The authors would like to thank the anonymous reviewers and the member of the Editorial Board for their useful and constructive comments.

\bibliographystyle{abbrvnat_brief}

\bibliography{Bibliography} 

\begin{thebibliography}{72}
\providecommand{\natexlab}[1]{#1}
\providecommand{\url}[1]{\texttt{#1}}
\expandafter\ifx\csname urlstyle\endcsname\relax
  \providecommand{\doi}[1]{doi: #1}\else
  \providecommand{\doi}{doi: \begingroup \urlstyle{rm}\Url}\fi

\bibitem[Adorf(1995)]{Adorf1995}
H.~M. Adorf.
\newblock {Hubble Space Telescope image restoration in its fourth year}.
\newblock \emph{Inverse Problems}, 11\penalty0 (4):\penalty0 639--653, 1995.

\bibitem[Amini and Wainwright(2012)]{Amini2012}
A.~A. Amini and M.~J. Wainwright.
\newblock {Sampled forms of functional PCA in reproducing kernel Hilbert
  spaces}.
\newblock \emph{The Annals of Statistics}, 40\penalty0 (5):\penalty0
  2483--2510, 2012.

\bibitem[Arridge(1999)]{Arridge1999}
S.~R. Arridge.
\newblock {Optical tomography in medical imaging}.
\newblock \emph{Inverse Problems}, 15\penalty0 (2):\penalty0 R41--R93, 1999.

\bibitem[Arridge et~al.(2006)Arridge, Kaipio, Kolehmainen, Schweiger,
  Somersalo, Tarvainen, and Vauhkonen]{Arridge2006}
S.~R. Arridge, J.~P. Kaipio, V.~Kolehmainen, M.~Schweiger, E.~Somersalo,
  T.~Tarvainen, and M.~Vauhkonen.
\newblock {Approximation errors and model reduction with an application in
  optical diffusion tomography}.
\newblock \emph{Inverse Problems}, 22\penalty0 (1):\penalty0 175--195, 2006.

\bibitem[Azzimonti et~al.(2014)Azzimonti, Nobile, Sangalli, and
  Secchi]{Azzimonti2014}
L.~Azzimonti, F.~Nobile, L.~M. Sangalli, and P.~Secchi.
\newblock {Mixed Finite Elements for Spatial Regression with PDE Penalization}.
\newblock \emph{SIAM/ASA Journal on Uncertainty Quantification}, 2\penalty0
  (1):\penalty0 305--335, 2014.

\bibitem[Beck and Teboulle(2009)]{Beck2009}
A.~Beck and M.~Teboulle.
\newblock {A Fast Iterative Shrinkage-Thresholding Algorithm for Linear Inverse
  Problems}.
\newblock \emph{SIAM Journal on Imaging Sciences}, 2\penalty0 (1):\penalty0
  183--202, 2009.

\bibitem[Benko et~al.(2009)Benko, H{\"{a}}rdle, and Kneip]{Benko2009}
M.~Benko, W.~H{\"{a}}rdle, and A.~Kneip.
\newblock {Common functional principal components}.
\newblock \emph{The Annals of Statistics}, 37\penalty0 (1):\penalty0 1--34,
  2009.

\bibitem[Boyd et~al.(2010)Boyd, Parikh, Chu, Peleato, and Eckstein]{Boyd2010}
S.~Boyd, N.~Parikh, E.~Chu, B.~Peleato, and J.~Eckstein.
\newblock {Distributed optimization and statistical learning via the
  alternating direction method of multipliers}.
\newblock \emph{Foundations and Trends in Machine Learning}, 3\penalty0
  (1):\penalty0 1--122, 2010.

\bibitem[Braess(2007)]{Braess2007}
D.~Braess.
\newblock \emph{{Finite Elements: Theory, Fast Solvers, and Applications in
  Solid Mechanics}}.
\newblock Cambridge University Press, Cambridge, 2007.
\newblock ISBN 9780511618635.

\bibitem[Buckner et~al.(2008)Buckner, Andrews-Hanna, and Schacter]{Buckner2008}
R.~L. Buckner, J.~R. Andrews-Hanna, and D.~L. Schacter.
\newblock {The Brain's Default Network}.
\newblock \emph{Annals of the New York Academy of Sciences}, 1124\penalty0
  (1):\penalty0 1--38, 2008.

\bibitem[Bunea and Xiao(2015)]{bunea2015}
F.~Bunea and L.~Xiao.
\newblock {On the sample covariance matrix estimator of reduced effective rank
  population matrices, with applications to fPCA}.
\newblock \emph{Bernoulli}, 21\penalty0 (2):\penalty0 1200--1230, 2015.

\bibitem[Burger et~al.(2016)Burger, Sawatzky, and Steidl]{Burger2016}
M.~Burger, A.~Sawatzky, and G.~Steidl.
\newblock {First Order Algorithms in Variational Image Processing}.
\newblock In R.~Glowinski, S.~Osher, and W.~Yin, editors, \emph{Splitting
  Methods in Communication, Imaging, Science, and Engineering}, pages 345--407.
  Springer, Cham, 2016.

\bibitem[Calder{\'{o}}n(1980)]{Calderon1980}
A.~Calder{\'{o}}n.
\newblock {On an inverse boundary value problem}.
\newblock In W.~H. Meyer and M.~A. Raupp, editors, \emph{Seminar on Numerical
  Analysis and its Applications to Continuum Physics}, pages 65--73. Soc.
  Brasileira de Matematica, Rio de Janeiro, 1980.

\bibitem[Calvetti and Somersalo(2007)]{Calvetti2007}
D.~Calvetti and E.~Somersalo.
\newblock \emph{{Introduction to Bayesian Scientific Computing: Ten Lectures on
  Subjective Computing}}, volume~2 of \emph{Surveys and Tutorials in the
  Applied Mathematical Sciences}.
\newblock Springer New York, New York, NY, 2007.
\newblock ISBN 978-0-387-73393-7.

\bibitem[Cavalier(2008)]{Cavalier2008}
L.~Cavalier.
\newblock {Nonparametric statistical inverse problems}.
\newblock \emph{Inverse Problems}, 24\penalty0 (3):\penalty0 034004, 2008.

\bibitem[Chambolle and Pock(2011)]{Chambolle2011}
A.~Chambolle and T.~Pock.
\newblock {A First-Order Primal-Dual Algorithm for Convex Problems with
  Applications to Imaging}.
\newblock \emph{Journal of Mathematical Imaging and Vision}, 40\penalty0
  (1):\penalty0 120--145, 2011.

\bibitem[Chambolle and Pock(2016)]{Chambolle2016}
A.~Chambolle and T.~Pock.
\newblock {An introduction to continuous optimization for imaging}.
\newblock \emph{Acta Numerica}, 25:\penalty0 161--319, 2016.

\bibitem[Colclough et~al.(2016)Colclough, Woolrich, Tewarie, Brookes, Quinn,
  and Smith]{Colclough2016}
G.~Colclough, M.~Woolrich, P.~Tewarie, M.~Brookes, A.~Quinn, and S.~Smith.
\newblock {How reliable are MEG resting-state connectivity metrics?}
\newblock \emph{NeuroImage}, 138:\penalty0 284--293, 2016.

\bibitem[{De Vito} et~al.(2005){De Vito}, Rosasco, Caponnetto, {De Giovannini},
  Odone, Vito, and {De Vito}]{DeVito2005}
E.~{De Vito}, L.~Rosasco, A.~Caponnetto, U.~{De Giovannini}, F.~Odone, D.~Vito,
  and F.~O. {De Vito}.
\newblock {Learning from Examples as an Inverse Problem}.
\newblock \emph{Journal of Machine Learning Research}, 6:\penalty0 883--904,
  2005.

\bibitem[Dobriban et~al.(2017)Dobriban, Leeb, and Singer]{Dobriban2017}
E.~Dobriban, W.~Leeb, and A.~Singer.
\newblock {Optimal prediction in the linearly transformed spiked model}.
\newblock 2017.

\bibitem[Dryden et~al.(2009)Dryden, Koloydenko, and Zhou]{Dryden2009}
I.~L. Dryden, A.~Koloydenko, and D.~Zhou.
\newblock {Non-Euclidean statistics for covariance matrices, with applications
  to diffusion tensor imaging}.
\newblock \emph{The Annals of Applied Statistics}, 3\penalty0 (3):\penalty0
  1102--1123, 2009.

\bibitem[Dziuk and Elliott(2013)]{Dziuk2013}
G.~Dziuk and C.~M. Elliott.
\newblock {Finite element methods for surface PDEs}.
\newblock \emph{Acta Numerica}, 22\penalty0 (April):\penalty0 289--396, 2013.

\bibitem[Eggebrecht et~al.(2014)Eggebrecht, Ferradal, Robichaux-Viehoever,
  Hassanpour, Dehghani, Snyder, Hershey, and Culver]{Eggebrecht2014}
A.~T. Eggebrecht, S.~L. Ferradal, A.~Robichaux-Viehoever, M.~S. Hassanpour,
  H.~Dehghani, A.~Z. Snyder, T.~Hershey, and J.~P. Culver.
\newblock {Mapping distributed brain function and networks with diffuse optical
  tomography}.
\newblock \emph{Nature Photonics}, 8\penalty0 (6):\penalty0 448--454, 2014.

\bibitem[Ferrari and Quaresima(2012)]{Ferrari2012}
M.~Ferrari and V.~Quaresima.
\newblock {A brief review on the history of human functional near-infrared
  spectroscopy (fNIRS) development and fields of application}.
\newblock \emph{NeuroImage}, 63\penalty0 (2):\penalty0 921--935, 2012.

\bibitem[Flury(1984)]{Flury1984}
B.~N. Flury.
\newblock {Common Principal Components in k Groups}.
\newblock \emph{Journal of the American Statistical Association}, 79\penalty0
  (388):\penalty0 892--898, 1984.

\bibitem[Fransson et~al.(2011)Fransson, {\AA}den, Blennow, and
  Lagercrantz]{Fransson2011}
P.~Fransson, U.~{\AA}den, M.~Blennow, and H.~Lagercrantz.
\newblock {The Functional Architecture of the Infant Brain as Revealed by
  Resting-State fMRI}.
\newblock \emph{Cerebral Cortex}, 21\penalty0 (1):\penalty0 145--154, 2011.

\bibitem[Fried and Malkus(1975)]{Fried1975}
I.~Fried and D.~S. Malkus.
\newblock {Finite element mass matrix lumping by numerical integration with no
  convergence rate loss}.
\newblock \emph{International Journal of Solids and Structures}, 11\penalty0
  (4):\penalty0 461--466, 1975.

\bibitem[Friedman et~al.(2008)Friedman, Hastie, and Tibshirani]{Friedman2008}
J.~Friedman, T.~Hastie, and R.~Tibshirani.
\newblock {Sparse inverse covariance estimation with the graphical lasso}.
\newblock \emph{Biostatistics}, 9\penalty0 (3):\penalty0 432--441, 2008.

\bibitem[Geman(1990)]{Geman1990}
D.~Geman.
\newblock {Random fields and inverse problems in imaging}.
\newblock In P.-L. Hennequin, editor, \emph{{\'{E}}cole d'{\'{E}}t{\'{e}} de
  Probabilit{\'{e}}s de Saint-Flour XVIII - 1988. Lecture Notes in
  Mathematics}, pages 115--193. Springer, Berlin, Heidelberg, 1990.

\bibitem[Golub and van Loan(1980)]{Golub1980}
G.~H. Golub and C.~F. van Loan.
\newblock {An Analysis of the Total Least Squares Problem}.
\newblock \emph{SIAM Journal on Numerical Analysis}, 17\penalty0 (6):\penalty0
  883--893, 1980.

\bibitem[Gross et~al.(2013)Gross, Baillet, Barnes, Henson, Hillebrand, Jensen,
  Jerbi, Litvak, Maess, Oostenveld, Parkkonen, Taylor, van Wassenhove, Wibral,
  and Schoffelen]{Gross2013}
J.~Gross, S.~Baillet, G.~R. Barnes, R.~N. Henson, A.~Hillebrand, O.~Jensen,
  K.~Jerbi, V.~Litvak, B.~Maess, R.~Oostenveld, L.~Parkkonen, J.~R. Taylor,
  V.~van Wassenhove, M.~Wibral, and J.-M. Schoffelen.
\newblock {Good practice for conducting and reporting MEG research}.
\newblock \emph{NeuroImage}, 65:\penalty0 349--363, 2013.

\bibitem[Gutta et~al.(2019)Gutta, Bhatt, Kalva, Pramanik, and
  Yalavarthy]{Gutta2019}
S.~Gutta, M.~Bhatt, S.~K. Kalva, M.~Pramanik, and P.~K. Yalavarthy.
\newblock {Modeling Errors Compensation With Total Least Squares for Limited
  Data Photoacoustic Tomography}.
\newblock \emph{IEEE Journal of Selected Topics in Quantum Electronics},
  25\penalty0 (1):\penalty0 1--14, 2019.

\bibitem[Hansen(2000)]{Hansen2000}
P.~C. Hansen.
\newblock {The L-Curve and its Use in the Numerical Treatment of Inverse
  Problems}.
\newblock In \emph{Computational Inverse Problems in Electrocardiology, ed. P.
  Johnston, Advances in Computational Bioengineering}, pages 119--142. WIT
  Press, 2000.

\bibitem[Hartley and Zisserman(2004)]{Hartley2004}
R.~Hartley and A.~Zisserman.
\newblock \emph{{Multiple View Geometry in Computer Vision}}.
\newblock Cambridge University Press, Cambridge, 2004.
\newblock ISBN 9780511811685.

\bibitem[Hu and Jacob(2012)]{YueHu2012}
Y.~Hu and M.~Jacob.
\newblock {Higher Degree Total Variation (HDTV) Regularization for Image
  Recovery}.
\newblock \emph{IEEE Transactions on Image Processing}, 21\penalty0
  (5):\penalty0 2559--2571, 2012.

\bibitem[Huang et~al.(2008)Huang, Shen, and Buja]{huang2008}
J.~Z. Huang, H.~Shen, and A.~Buja.
\newblock {Functional principal components analysis via penalized rank one
  approximation}.
\newblock \emph{Electronic Journal of Statistics}, 2\penalty0 (March):\penalty0
  678--695, 2008.

\bibitem[Hutchison et~al.(2013)Hutchison, Womelsdorf, Allen, Bandettini,
  Calhoun, Corbetta, {Della Penna}, Duyn, Glover, Gonzalez-Castillo,
  Handwerker, Keilholz, Kiviniemi, Leopold, de~Pasquale, Sporns, Walter, and
  Chang]{Hutchison2013}
R.~M. Hutchison, T.~Womelsdorf, E.~A. Allen, P.~A. Bandettini, V.~D. Calhoun,
  M.~Corbetta, S.~{Della Penna}, J.~H. Duyn, G.~H. Glover,
  J.~Gonzalez-Castillo, D.~A. Handwerker, S.~Keilholz, V.~Kiviniemi, D.~A.
  Leopold, F.~de~Pasquale, O.~Sporns, M.~Walter, and C.~Chang.
\newblock {Dynamic functional connectivity: Promise, issues, and
  interpretations}.
\newblock \emph{NeuroImage}, 80:\penalty0 360--378, 2013.

\bibitem[Jin et~al.(2017)Jin, McCann, Froustey, and Unser]{Jin2017}
K.~H. Jin, M.~T. McCann, E.~Froustey, and M.~Unser.
\newblock {Deep Convolutional Neural Network for Inverse Problems in Imaging}.
\newblock \emph{IEEE Transactions on Image Processing}, 26\penalty0
  (9):\penalty0 4509--4522, 2017.

\bibitem[Johnstone and Silverman(1990)]{johnstone1990}
I.~M. Johnstone and B.~W. Silverman.
\newblock {Speed of Estimation in Positron Emission Tomography and Related
  Inverse Problems}.
\newblock \emph{The Annals of Statistics}, 18\penalty0 (1):\penalty0 251--280,
  1990.

\bibitem[Kaipio and Somersalo(2005)]{Kaipio2005}
J.~Kaipio and E.~Somersalo.
\newblock \emph{{Statistical and Computational Inverse Problems}}, volume 160
  of \emph{Applied Mathematical Sciences}.
\newblock Springer-Verlag, New York, 2005.
\newblock ISBN 0-387-22073-9.

\bibitem[Katsevich et~al.(2015)Katsevich, Katsevich, and Singer]{Katsevich2015}
E.~Katsevich, A.~Katsevich, and A.~Singer.
\newblock {Covariance Matrix Estimation for the Cryo-EM Heterogeneity Problem}.
\newblock \emph{SIAM Journal on Imaging Sciences}, 8\penalty0 (1):\penalty0
  126--185, 2015.

\bibitem[Kluth and Maass(2017)]{Kluth2017}
T.~Kluth and P.~Maass.
\newblock {Model uncertainty in magnetic particle imaging: Nonlinear problem
  formulation and model-based sparse reconstruction}.
\newblock \emph{International Journal on Magnetic Particle Imaging}, 3\penalty0
  (2), 2017.

\bibitem[Lee et~al.(2013)Lee, Zahneisen, Hugger, LeVan, and Hennig]{Lee2013}
H.-L. Lee, B.~Zahneisen, T.~Hugger, P.~LeVan, and J.~Hennig.
\newblock {Tracking dynamic resting-state networks at higher frequencies using
  MR-encephalography}.
\newblock \emph{NeuroImage}, 65:\penalty0 216--222, 2013.

\bibitem[Lefkimmiatis et~al.(2012)Lefkimmiatis, Bourquard, and
  Unser]{Lefkimmiatis2012}
S.~Lefkimmiatis, A.~Bourquard, and M.~Unser.
\newblock {Hessian-Based Norm Regularization for Image Restoration With
  Biomedical Applications}.
\newblock \emph{IEEE Transactions on Image Processing}, 21\penalty0
  (3):\penalty0 983--995, 2012.

\bibitem[Lehikoinen et~al.(2007)Lehikoinen, Finsterle, Voutilainen, Heikkinen,
  Vauhkonen, and Kaipio]{Lehikoinen2007}
A.~Lehikoinen, S.~Finsterle, A.~Voutilainen, L.~M. Heikkinen, M.~Vauhkonen, and
  J.~P. Kaipio.
\newblock {Approximation errors and truncation of computational domains with
  application to geophysical tomography}.
\newblock \emph{Inverse Problems and Imaging}, 1\penalty0 (2):\penalty0
  371--389, 2007.

\bibitem[Li et~al.(2009)Li, Guo, Nie, Li, and Liu]{Li2009}
K.~Li, L.~Guo, J.~Nie, G.~Li, and T.~Liu.
\newblock {Review of methods for functional brain connectivity detection using
  fMRI}.
\newblock \emph{Computerized Medical Imaging and Graphics}, 33\penalty0
  (2):\penalty0 131--139, 2009.

\bibitem[Lila et~al.(2016)Lila, Aston, and Sangalli]{Lila2016}
E.~Lila, J.~A.~D. Aston, and L.~M. Sangalli.
\newblock {Smooth Principal Component Analysis over two-dimensional manifolds
  with an application to neuroimaging}.
\newblock \emph{The Annals of Applied Statistics}, 10\penalty0 (4):\penalty0
  1854--1879, 2016.

\bibitem[Liu and Zhang(2019)]{Liu2019}
X.~Liu and N.~Zhang.
\newblock {Sparse Inverse Covariance Matrix Estimation via the l0-norm with
  Tikhonov Regularization}.
\newblock \emph{Inverse Problems}, 2019.

\bibitem[Lustig et~al.(2008)Lustig, Donoho, Santos, and Pauly]{Lustig2008}
M.~Lustig, D.~Donoho, J.~Santos, and J.~Pauly.
\newblock {Compressed Sensing MRI}.
\newblock \emph{IEEE Signal Processing Magazine}, 25\penalty0 (2):\penalty0
  72--82, 2008.

\bibitem[Math{\'{e}} and Pereverzev(2006)]{Mathe2006}
P.~Math{\'{e}} and S.~V. Pereverzev.
\newblock {Regularization of some linear ill-posed problems with discretized
  random noisy data}.
\newblock \emph{Mathematics of Computation}, 75\penalty0 (256):\penalty0
  1913--1929, 2006.

\bibitem[McCann et~al.(2017)McCann, Jin, and Unser]{McCann2017}
M.~T. McCann, K.~H. Jin, and M.~Unser.
\newblock {Convolutional Neural Networks for Inverse Problems in Imaging: A
  Review}.
\newblock \emph{IEEE Signal Processing Magazine}, 34\penalty0 (6):\penalty0
  85--95, 2017.

\bibitem[Mosher et~al.(1999)Mosher, Leahy, and Lewis]{Mosher1999}
J.~Mosher, R.~Leahy, and P.~Lewis.
\newblock {EEG and MEG: forward solutions for inverse methods}.
\newblock \emph{IEEE Transactions on Biomedical Engineering}, 46\penalty0
  (3):\penalty0 245--259, 1999.

\bibitem[Nissinen et~al.(2009)Nissinen, Heikkinen, Kolehmainen, and
  Kaipio]{Nissinen2009}
A.~Nissinen, L.~M. Heikkinen, V.~Kolehmainen, and J.~P. Kaipio.
\newblock {Compensation of errors due to discretization, domain truncation and
  unknown contact impedances in electrical impedance tomography}.
\newblock \emph{Measurement Science and Technology}, 20\penalty0 (10):\penalty0
  105504, 2009.

\bibitem[Oostenveld et~al.(2011)Oostenveld, Fries, Maris, and
  Schoffelen]{Oostenveld2011}
R.~Oostenveld, P.~Fries, E.~Maris, and J.-M. Schoffelen.
\newblock {FieldTrip: Open Source Software for Advanced Analysis of MEG, EEG,
  and Invasive Electrophysiological Data}.
\newblock \emph{Computational Intelligence and Neuroscience}, 2011:\penalty0
  1--9, 2011.

\bibitem[Petersen and M{\"{u}}ller(2019)]{Petersen2019}
A.~Petersen and H.-G. M{\"{u}}ller.
\newblock {Fr{\'{e}}chet regression for random objects with Euclidean
  predictors}.
\newblock \emph{The Annals of Statistics}, 47\penalty0 (2):\penalty0 691--719,
  2019.

\bibitem[Pigoli et~al.(2014)Pigoli, Aston, Dryden, and Secchi]{Pigoli2014}
D.~Pigoli, J.~A. Aston, I.~L. Dryden, and P.~Secchi.
\newblock {Distances and inference for covariance operators}.
\newblock \emph{Biometrika}, 101\penalty0 (2):\penalty0 409--422, 2014.

\bibitem[Ramsay and Silverman(2005)]{ramsay2005}
J.~Ramsay and W.~B. Silverman.
\newblock \emph{{Functional Data Analysis}}.
\newblock Springer Series in Statistics. Springer-Verlag, New York, 2005.
\newblock ISBN 0-387-40080-X.

\bibitem[Repetti et~al.(2019)Repetti, Pereyra, and Wiaux]{Repetti2019}
A.~Repetti, M.~Pereyra, and Y.~Wiaux.
\newblock {Scalable Bayesian Uncertainty Quantification in Imaging Inverse
  Problems via Convex Optimization}.
\newblock \emph{SIAM Journal on Imaging Sciences}, 12\penalty0 (1):\penalty0
  87--118, 2019.

\bibitem[Riesz and Szokefalvi-Nagy(1955)]{riesz1955}
F.~Riesz and B.~Szokefalvi-Nagy.
\newblock \emph{{Functional analysis}}.
\newblock F. Ungar Pub. Co., New York, 1955.

\bibitem[Silverman(1996)]{silverman1996}
B.~W. Silverman.
\newblock {Smoothed functional principal components analysis by choice of
  norm}.
\newblock \emph{The Annals of Statistics}, 24\penalty0 (1):\penalty0 1--24,
  1996.

\bibitem[Singh et~al.(2014)Singh, Cooper, {Wai Lee}, Dempsey, Edwards,
  Brigadoi, Airantzis, Everdell, Michell, Holder, Hebden, and
  Austin]{Singh2014}
H.~Singh, R.~J. Cooper, C.~{Wai Lee}, L.~Dempsey, A.~Edwards, S.~Brigadoi,
  D.~Airantzis, N.~Everdell, A.~Michell, D.~Holder, J.~C. Hebden, and
  T.~Austin.
\newblock {Mapping cortical haemodynamics during neonatal seizures using
  diffuse optical tomography: A case study}.
\newblock \emph{NeuroImage: Clinical}, 5:\penalty0 256--265, 2014.

\bibitem[Stuart(2010)]{Stuart2010}
A.~M. Stuart.
\newblock {Inverse problems: A Bayesian perspective}.
\newblock \emph{Acta Numerica}, 19\penalty0 (2010):\penalty0 451--559, 2010.

\bibitem[Tenorio(2001)]{Tenorio2001}
L.~Tenorio.
\newblock {Statistical Regularization of Inverse Problems}.
\newblock \emph{SIAM Review}, 43\penalty0 (2):\penalty0 347--366, 2001.

\bibitem[Tian et~al.(2012)Tian, Huang, Shen, and Li]{Tian2012}
T.~S. Tian, J.~Z. Huang, H.~Shen, and Z.~Li.
\newblock {A two-way regularization method for MEG source reconstruction}.
\newblock \emph{The Annals of Applied Statistics}, 6\penalty0 (3):\penalty0
  1021--1046, 2012.

\bibitem[{Van Essen} et~al.(2012){Van Essen}, Ugurbil, Auerbach, Barch,
  Behrens, Bucholz, Chang, Chen, Corbetta, Curtiss, {Della Penna}, Feinberg,
  Glasser, Harel, Heath, Larson-Prior, Marcus, Michalareas, Moeller,
  Oostenveld, Petersen, Prior, Schlaggar, Smith, Snyder, Xu, and
  Yacoub]{Essen2012}
D.~{Van Essen}, K.~Ugurbil, E.~Auerbach, D.~Barch, T.~Behrens, R.~Bucholz,
  A.~Chang, L.~Chen, M.~Corbetta, S.~Curtiss, S.~{Della Penna}, D.~Feinberg,
  M.~Glasser, N.~Harel, A.~Heath, L.~Larson-Prior, D.~Marcus, G.~Michalareas,
  S.~Moeller, R.~Oostenveld, S.~Petersen, F.~Prior, B.~Schlaggar, S.~Smith,
  A.~Snyder, J.~Xu, and E.~Yacoub.
\newblock {The Human Connectome Project: A data acquisition perspective}.
\newblock \emph{NeuroImage}, 62\penalty0 (4):\penalty0 2222--2231, 2012.

\bibitem[Vogel(2002)]{Vogel2002}
C.~R. Vogel.
\newblock \emph{{Computational Methods for Inverse Problems}}.
\newblock Society for Industrial and Applied Mathematics, 2002.
\newblock ISBN 978-0-89871-550-7.

\bibitem[Yao et~al.(2005)Yao, M{\"{u}}ller, and Wang]{Yao2005}
F.~Yao, H.-G. M{\"{u}}ller, and J.-L. Wang.
\newblock {Functional Data Analysis for Sparse Longitudinal Data}.
\newblock \emph{Journal of the American Statistical Association}, 100\penalty0
  (470):\penalty0 577--590, 2005.

\bibitem[Ye et~al.(2009)Ye, Tak, Jang, Jung, and Jang]{Ye2009}
J.~C. Ye, S.~Tak, K.~E. Jang, J.~Jung, and J.~Jang.
\newblock {NIRS-SPM: Statistical parametric mapping for near-infrared
  spectroscopy}.
\newblock \emph{NeuroImage}, 44\penalty0 (2):\penalty0 428--447, 2009.

\bibitem[Yeo et~al.(2011)Yeo, Krienen, Sepulcre, Sabuncu, Lashkari,
  Hollinshead, Roffman, Smoller, Z{\"{o}}llei, Polimeni, Fischl, Liu, and
  Buckner]{Yeo2011}
B.~T.~T. Yeo, F.~M. Krienen, J.~Sepulcre, M.~R. Sabuncu, D.~Lashkari,
  M.~Hollinshead, J.~L. Roffman, J.~W. Smoller, L.~Z{\"{o}}llei, J.~R.
  Polimeni, B.~Fischl, H.~Liu, and R.~L. Buckner.
\newblock {The organization of the human cerebral cortex estimated by intrinsic
  functional connectivity}.
\newblock \emph{Journal of Neurophysiology}, 106\penalty0 (3):\penalty0
  1125--1165, 2011.

\bibitem[Zhdanov(2002)]{Zhdanov2002}
M.~Zhdanov.
\newblock \emph{{Inverse Theory and Applications in Geophysics}}.
\newblock Elsevier, 2002.
\newblock ISBN 9780444627124.

\bibitem[Zhu et~al.(2011)Zhu, Leus, and Giannakis]{Zhu2011}
H.~Zhu, G.~Leus, and G.~B. Giannakis.
\newblock {Sparsity-Cognizant Total Least-Squares for Perturbed Compressive
  Sampling}.
\newblock \emph{IEEE Transactions on Signal Processing}, 59\penalty0
  (5):\penalty0 2002--2016, 2011.

\bibitem[Zienkiewicz et~al.(2013)Zienkiewicz, Taylor, and Zhu]{Zienkiewicz2013}
O.~Zienkiewicz, R.~Taylor, and J.~Z. Zhu.
\newblock \emph{{The Finite Element Method: its Basis and Fundamentals: Seventh
  Edition}}.
\newblock Elsevier, 2013.
\newblock ISBN 9781856176330.

\end{thebibliography}

\end{document}